\documentclass{llncs}
\usepackage{times}
\usepackage{amsmath}
\usepackage{amsfonts}
\usepackage{amssymb}
\usepackage{graphicx}
\usepackage{listings}
\usepackage{paralist}
\usepackage{tikz}
\usepackage{hyperref}
\usepackage{url}
\usepackage{multicol}
\usepackage{wrapfig}
\usepackage{array}
\usepackage{color}
 \usepackage{caption}
\usepackage{booktabs}
\usepackage[font=small,skip=0pt]{caption}
\usepackage{subcaption}
\usepackage{enumerate}
\usepackage{algorithm}
\usepackage[noend]{algpseudocode}
\errorcontextlines\maxdimen

\usepackage{multicol}

\algnewcommand{\IfThenElse}[3]{% \IfThenElse{<if>}{<then>}{<else>}
	\State \algorithmicif\ #1\ \algorithmicthen\ #2\ \algorithmicelse\ #3}

\newcommand{\chuchu}[1]{\textcolor{black}{#1}}

\newcommand{\RT}{R}

\newcommand{\EA}{EA}
\newcommand{\bloatf}{\mathit{ted}}
\newcommand{\Post}{\mathit{CompTed}}
\newcommand{\eep}{\mathit{eep}}

\hypersetup{
    bookmarks=true,         % show bookmarks bar?
    unicode=false,          % non-Latin characters in Acrobat’s bookmarks
    pdftoolbar=true,        % show Acrobat’s toolbar?
    pdfmenubar=true,        % show Acrobat’s menu?
    pdffitwindow=false,     % window fit to page when opened
    pdfstartview={FitH},    % fits the width of the page to the window
    pdftitle={My title},    % title
    pdfauthor={Author},     % author
    pdfsubject={Subject},   % subject of the document
    pdfcreator={Creator},   % creator of the document
    pdfproducer={Producer}, % producer of the document
    pdfkeywords={keyword1} {key2} {key3}, % list of keywords
    pdfnewwindow=true,      % links in new window
    colorlinks=true,       % false: boxed links; true: colored links
    linkcolor=blue,          % color of internal links (change box color with linkbordercolor)
    citecolor=blue,        % color of links to bibliography
    filecolor=magenta,      % color of file links
    urlcolor=blue           % color of external links
}

\newtheorem{assumption}{Assumption}
\def\A{{\cal A}} % HA
\def\B{{\cal B}} % HA
\def\E{{\cal E}} % HA
\def\I{{\cal I}} % environment sequence
\def\N{{\cal N}} % Mode switching transitions
\def\P{{\cal P}} % set of modes
\def\R{{\cal R}} % relation
\def\S{{\cal S}} % set of trajectories
\def\U{{\cal U}} % set of trajectories
\newcommand{\val}[1]{\relax\ifmmode {\mbox{\it Val}}(#1) \else ${{\it Val}}({#1})$\fi} 
\newcommand{\traj}[1]{\relax\ifmmode {\mathit Traj}(#1) \else ${\mathit Traj}({#1})$\fi}

\newcommand{\arrow}[1]{\mathrel{\stackrel{#1}{\rightarrow}}}

\newcommand{\deq}{\mathrel{\stackrel{\scriptscriptstyle\Delta}{=}}}

\newcommand{\set}[2]{\{#1 \ | \ #2\}}

\newcommand{\ball}[2]{\B_{#1}(#2)}
\newcommand{\restr}{\mathrel{\lceil}}

\newcommand{\mysim}[1]{\stackrel{#1}{\sim}}
\newcommand{\myequiv}[1]{\stackrel{#1}{\equiv}}
\newcommand{\myequivexe}[1]{\stackrel{#1}{\approx}}
\newcommand{\trace}[1]{\mathit{trace}(#1)}
\newcommand{\exe}[1]{\xi_{#1}}

\newcommand{\execs}[1]{\relax\ifmmode {\sf Execs}(#1) \else ${\sf Exec}(#1)$\fi} 
\newcommand{\traces}[1]{\relax\ifmmode {\sf Traces}(#1) \else ${\sf Tracec}(#1)$\fi}
\newcommand{\reach}[2]{\relax\ifmmode {\sf Reach}_{#1}(#2) \else ${\sf Reach}_{#1}(#2)$\fi} 

\newcommand{\frag}[1]{\relax\ifmmode {\sf Frags}_{#1} \else ${\sf Frags}_{#1}$\fi} 
\newcommand{\fragf}[1]{\relax\ifmmode {\sf Frags}^*_{#1} \else ${\sf Frags}^*_{#1}$\fi} 
\newcommand{\fragi}[1]{\relax\ifmmode {\sf Frags}^\omega_{#1} \else ${\sf Frags}^\omega_{#1}$\fi}

\newcommand{\num}[1]{\relax\ifmmode \mathbb #1\else $\mathbb #1$\fi}
\newcommand{\nnnum}[1]{\relax\ifmmode 
	{\mathbb #1}_{\geq 0} \else ${\mathbb #1}_{\geq 0}$
	\fi}
\newcommand{\npnum}[1]{\relax\ifmmode 
	{\mathbb #1}_{\leq 0} \else ${\mathbb #1}_{\leq 0}$
	\fi}
\newcommand{\pnum}[1]{\relax\ifmmode 
	{\mathbb #1}_{> 0} \else ${\mathbb #1}_{> 0}$
	\fi}
\newcommand{\nnum}[1]{\relax\ifmmode 
	{\mathbb #1}_{< 0} \else ${\mathbb #1}_{< 0}$
	\fi}
\newcommand{\plnum}[1]{\relax\ifmmode 
	{\mathbb #1}_{+} \else ${\mathbb #1}_{+}$
	\fi}
\newcommand{\nenum}[1]{\relax\ifmmode 
	{\mathbb #1}_{-} \else ${\mathbb #1}_{-}$
	\fi}

\newcommand{\reals}{{\num R}}                    %reals
\newcommand{\booleans}{{\num B}}                    %reals
\newcommand{\nnreals}{{\nnnum R}}                    %nonnegative reals
\newcommand{\naturals}{{\num N}}                      %natural numbers
\newcommand{\fstate}{\mathop{\mathsf {fstate}}}
\newcommand{\lstate}{\mathop{\mathsf {lstate}}}

\newcommand{\auto}[1]{{\operatorname{\mathsf{#1}}}}
\newcommand{\act}[1]{{\operatorname{\mathsf{#1}}}}

\newcommand{\two}[4]{
	\parbox{.95\columnwidth}{\vspace{1pt} \vfill
		\parbox[t]{#1\columnwidth}{#3}%
		\parbox[t]{#2\columnwidth}{#4}%
	}}

				\newcommand{\tup}[1]
				{
					\relax\ifmmode
					\langle #1 \rangle
					\else $\langle$ #1 $\rangle$ \fi
				}
				
				\newcommand{\lit}[1]{ \relax\ifmmode
					\mathord{\mathcode`\-="702D\sf #1\mathcode`\-="2200}
					\else {\it #1} \fi }

				\newcommand{\figuresize}{\scriptsize}

				\lstdefinelanguage{ioa}{
					basicstyle=\figuresize,
					keywordstyle=\bf \figuresize,
					identifierstyle=\it \figuresize,
					emphstyle=\tt \figuresize,
					mathescape=true,
					tabsize=20,
					sensitive=false,
					columns=fullflexible,
					keepspaces=false,
					flexiblecolumns=true,
					basewidth=0.05em,
					escapeinside={(*@}{@*)},
					moredelim=[il][\rm]{//},
					moredelim=[is][\sf \figuresize]{!}{!},
					moredelim=[is][\bf \figuresize]{*}{*},
					keywords={automaton, automata, and, assumed,
						choose,const,continue, components,
						derived, discrete, do,
						eff, external,else, elseif, evolve, end,
						fi,for, each, foreach, forward, from,
						hidden,
						in,input,internal,if,invariant, initially, imports,
						let,
						or, output, operators, od, of,
						pre, private,
						return,
						satisfies, shared, signature, simulation, stop, such,constants,
						trajectories,trajdef, transitions, that,then, type, types, to, tasks,
						variables, vocabulary, 
						when,where, with,while},
					emph={Set, set, seq, tuple, map, array, enumeration},   
					literate=
					{(}{{$($}}1
					{)}{{$)$}}1
					{\\in}{{$\in\ $}}1
					{\\preceq}{{$\preceq\ $}}1
					{\\subset}{{$\subset\ $}}1
					{\\subseteq}{{$\subseteq\ $}}1
					{\\supset}{{$\supset\ $}}1
					{\\supseteq}{{$\supseteq\ $}}1
					{\\forall}{{$\forall$}}1
					{\\le}{{$\le\ $}}1
					{\\ge}{{$\ge\ $}}1
					{\\gets}{{$\gets\ $}}1
					{\\cup}{{$\cup\ $}}1
					{\\cap}{{$\cap\ $}}1
					{\\langle}{{$\langle$}}1
					{\\rangle}{{$\rangle$}}1
					{\\exists}{{$\exists\ $}}1
					{\\bot}{{$\bot$}}1
					{\\rip}{{$\rip$}}1
					{\\emptyset}{{$\emptyset$}}1
					{\\notin}{{$\notin\ $}}1
					{\\not\\exists}{{$\not\exists\ $}}1
					{\\ne}{{$\ne\ $}}1
					{\\to}{{$\to\ $}}1
					{\\implies}{{$\implies\ $}}1
					{<}{{$<\ $}}1
					{>}{{$>\ $}}1
					{=}{{$=\ $}}1
					{~}{{$\neg\ $}}1
					{|}{{$\mid$}}1
					{'}{{$^\prime$}}1
					{\\A}{{$\forall\ $}}1
					{\\E}{{$\exists\ $}}1
					{\\nE}{{$\nexists\ $}}1
					{\\/}{{$\vee\,$}}1
					{\\vee}{{$\vee\,$}}1
					{/\\}{{$\wedge\,$}}1
					{\\wedge}{{$\wedge\,$}}1
					{=>}{{$\Rightarrow\ $}}1
					{->}{{$\rightarrow\ $}}1
					{<=}{{$\Leftarrow\ $}}1
					{<-}{{$\leftarrow\ $}}1
					{~=}{{$\neq\ $}}1
					{\\U}{{$\cup\ $}}1
					{\\I}{{$\cap\ $}}1
					{|-}{{$\vdash\ $}}1
					{-|}{{$\dashv\ $}}1
					{<<}{{$\ll\ $}}2
					{>>}{{$\gg\ $}}2
					{|}{{$\|$}}1
					{[}{{$[$}}1
					{]}{{$\,]$}}1
					{[[}{{$\langle$}}1
					{]]]}{{$]\rangle$}}1
					{]]}{{$\rangle$}}1
					{<=>}{{$\Leftrightarrow\ $}}2
					{<->}{{$\leftrightarrow\ $}}2
					{(+)}{{$\oplus\ $}}1
					{(-)}{{$\ominus\ $}}1
					{_i}{{$_{i}$}}1
					{_j}{{$_{j}$}}1
					{_{i,j}}{{$_{i,j}$}}3
					{_{j,i}}{{$_{j,i}$}}3
					{_0}{{$_0$}}1
					{_1}{{$_1$}}1
					{_2}{{$_2$}}1
					{_n}{{$_n$}}1
					{_p}{{$_p$}}1
					{_k}{{$_n$}}1
					{-}{{$\ms{-}$}}1
					{@}{{}}0
					{\\delta}{{$\delta$}}1
					{\\R}{{$\R$}}1
					{\\Rplus}{{$\Rplus$}}1
					{\\N}{{$\N$}}1
					{\\times}{{$\times\ $}}1
					{\\tau}{{$\tau$}}1
					{\\alpha}{{$\alpha$}}1
					{\\beta}{{$\beta$}}1
					{\\gamma}{{$\gamma$}}1
					{\\ell}{{$\ell\ $}}1
					{--}{{$-\ $}}1
					{\\TT}{{\hspace{1.5em}}}3        
				}
				
				\lstdefinelanguage{ioaNums}[]{ioa}
				{
					numbers=left,
					numberstyle=\tiny,
					stepnumber=2,
					numbersep=4pt
				}
				
				\lstdefinelanguage{ioaNumsRight}[]{ioa}
				{
					numbers=right,
					numberstyle=\tiny,
					stepnumber=2,
					numbersep=4pt
				}

				\lstnewenvironment{IOA}%
				{\lstset{language=IOA}}
				{}
				
				\lstnewenvironment{IOANums}%
				{
					\if@firstcolumn
					\lstset{language=IOA, numbers=left, firstnumber=auto}
					\else
					\lstset{language=IOA, numbers=right, firstnumber=auto}
					\fi
				}
				{}
				
				\lstnewenvironment{IOANumsRight}%
				{
					\lstset{language=IOA, numbers=right, firstnumber=auto}
				}
				{}
				
				\newcommand{\linefigioa}[9]{
					
				}

				\lstdefinelanguage{ioaLang}{%
					basicstyle=\ttfamily\small,
					keywordstyle=\rmfamily\bfseries\small,
					identifierstyle=\small,
					keywords={assumes,automaton,axioms,backward,bounds,by,case,choose,components,const,d,det,discrete,do,eff,else,elseif,ensuring,enumeration,evolve,fi,fire,follow,for,forward,from,hidden,if,in,%
						input,initially,internal,invariant,let, local,od,of,output,pre,schedule,signature,so,%
						simulation,states,variables, tasks, stop,tasks,that,then,to,trajdef,trajectory,trajectories,transitions,tuple,type,union,urgent,uses,when,where,while,yield},
					literate=
					{\\in}{{$\in$}}1
					{\\preceq}{{$\preceq$}}1
					{\\subset}{{$\subset$}}1
					{\\subseteq}{{$\subseteq$}}1
					{\\supset}{{$\supset$}}1
					{\\supseteq}{{$\supseteq$}}1
					{\\rho}{{$\rho$}}1
					{\\infty}{{$\infty$}}1
					{<}{{$<$}}1
					{>}{{$>$}}1
					{=}{{$=$}}1
					{~}{{$\neg$}}1 
					{|}{{$\mid$}}1
					{'}{{$^\prime$}}1
					{\\A}{{$\forall$}}1 {\\E}{{$\exists$}}1
					{\\/}{{$\vee$}}1 {/\\}{{$\wedge$}}1 
					{=>}{{$\Rightarrow$}}1 
					{->}{{$\rightarrow$}}1 
					{<=}{{$\leq$}}1 {>=}{{$\geq$}}1 {~=}{{$\neq$}}1
					{\\U}{{$\cup$}}1 {\\I}{{$\cap$}}1
					{|-}{{$\vdash$}}1 {-|}{{$\dashv$}}1
					{<<}{{$\ll$}}2 {>>}{{$\gg$}}2
					{|}{{$\|$}}1
					{<=>}{{$\Leftrightarrow$}}2 
					{<->}{{$\leftrightarrow$}}2
					{(+)}{{$\oplus$}}1
					{(-)}{{$\ominus$}}1
				}
				
				\lstdefinelanguage{bigIOALang}{%
					basicstyle=\ttfamily,
					keywordstyle=\rmfamily\bfseries,
					identifierstyle=,
					keywords={assumes,automaton,axioms,backward,by,case,choose,components,const,%
						d,det,discrete,do,eff,else,elseif,ensuring,enumeration,evolve,fi,for,forward,from,hidden,if,in%
						input,initially,internal,invariant,local,od,of,output,pre,schedule,signature,so,%
						tasks, simulation,states,stop,tasks,that,then,to,trajdef,trajectories,transitions,tuple,type,union,urgent,uses,when,where,yield},
					literate=
					{\\in}{{$\in$}}1
					{\\preceq}{{$\preceq$}}1
					{\\subset}{{$\subset$}}1
					{\\subseteq}{{$\subseteq$}}1
					{\\supset}{{$\supset$}}1
					{\\supseteq}{{$\supseteq$}}1
					{<}{{$<$}}1
					{>}{{$>$}}1
					{=}{{$=$}}1
					{~}{{$\neg$}}1 
					{|}{{$\mid$}}1
					{'}{{$^\prime$}}1
					{\\A}{{$\forall$}}1 {\\E}{{$\exists$}}1
					{\\/}{{$\vee$}}1 {/\\}{{$\wedge$}}1 
					{=>}{{$\Rightarrow$}}1 
					{->}{{$\rightarrow$}}1 
					{<=}{{$\leq$}}1 {>=}{{$\geq$}}1 {~=}{{$\neq$}}1
					{\\U}{{$\cup$}}1 {\\I}{{$\cap$}}1
					{|-}{{$\vdash$}}1 {-|}{{$\dashv$}}1
					{<<}{{$\ll$}}2 {>>}{{$\gg$}}2
					{|}{{$\|$}}1
					{<=>}{{$\Leftrightarrow$}}2 
					{<->}{{$\leftrightarrow$}}2
					{(+)}{{$\oplus$}}1
					{(-)}{{$\ominus$}}1
				}

				\lstnewenvironment{BigIOA}%
				{\lstset{language=bigIOALang,basicstyle=\ttfamily}
					\csname lst@SetFirstLabel\endcsname}
				{\csname lst@SaveFirstLabel\endcsname\vspace{-4pt}\noindent}
				
				\lstnewenvironment{SmallIOA}%
				{\lstset{language=ioaLang,basicstyle=\ttfamily\scriptsize}
					\csname lst@SetFirstLabel\endcsname}
				{\csname lst@SaveFirstLabel\endcsname\noindent}

				\newcommand{\true}{\relax\ifmmode \mathit true \else \em true \/\fi}
				\newcommand{\false}{\relax\ifmmode \mathit false \else \em false \/\fi}

				\newcommand{\Real}{{\operatorname{\texttt{Real}}}}
				\newcommand{\Bool}{{\operatorname{\texttt{Bool}}}}

				\newcommand{\ioaNat}{{\operatorname{\texttt{Nat}}}}

				\newlength{\bracklen}

				\newcommand{\tri}[3]{\ensuremath{\mathit{#1}^\mathit{#2}_\mathit{#3}}}

				\newcommand{\sugLocalVars}[2]{\ifthenelse{\equal{}{#2}}%
					{\tri{localVars}{#1}{desug}}%
					{\tri{localVars}{#1}{#2,desug}}}
				\newcommand{\sugVars}[2]{\ifthenelse{\equal{}{#2}}%
					{\tri{vars}{#1}{desug}}%
					{\tri{vars}{#1}{#2,desug}}}

				\newenvironment{subSyntax}{\begin{array}{l}}{\end{array}}
				\newcommand{\ms}[1]{\ifmmode%
					\mathord{\mathcode`-="702D\it #1\mathcode`\-="2200}\else%
					$\mathord{\mathcode`-="702D\it #1\mathcode`\-="2200}$\fi}
				
				\def\A{{\cal A}} % TA
				\def\B{{\cal B}} % TA
				\newcommand{\vv}{{\bf v}}

				\lstdefinelanguage{pvs}{
					basicstyle=\tt \figuresize,
					keywordstyle=\sc \figuresize,
					identifierstyle=\it \figuresize,
					emphstyle=\tt \figuresize,
					mathescape=true,
					tabsize=20,
					sensitive=false,
					columns=fullflexible,
					keepspaces=false,
					flexiblecolumns=true,
					basewidth=0.05em,
					moredelim=[il][\rm]{//},
					moredelim=[is][\sf \figuresize]{!}{!},
					moredelim=[is][\bf \figuresize]{*}{*},
					keywords={and, 
						begin,
						cases, const,
						do,
						external, else, exists, end, endcases, endif,
						fi,for, forall, from,
						hidden,
						in, if, importing,
						let, lambda, lemma,
						measure, 
						not,
						or, of,
						return, recursive,
						stop, 
						theory, that,then, type, types, type+, to, theorem,
						var,
						with,while},
					emph={nat, setof, sequence, eq, tuple, map, array, enumeration, bool, real, exp, nnreal, posreal},   
					literate=
					{(}{{$($}}1
					{)}{{$)$}}1
					{\\in}{{$\in\ $}}1
					{\\mapsto}{{$\rightarrow\ $}}1
					{\\preceq}{{$\preceq\ $}}1
					{\\subset}{{$\subset\ $}}1
					{\\subseteq}{{$\subseteq\ $}}1
					{\\supset}{{$\supset\ $}}1
					{\\supseteq}{{$\supseteq\ $}}1
					{\\forall}{{$\forall$}}1
					{\\le}{{$\le\ $}}1
					{\\ge}{{$\ge\ $}}1
					{\\gets}{{$\gets\ $}}1
					{\\cup}{{$\cup\ $}}1
					{\\cap}{{$\cap\ $}}1
					{\\langle}{{$\langle$}}1
					{\\rangle}{{$\rangle$}}1
					{\\exists}{{$\exists\ $}}1
					{\\bot}{{$\bot$}}1
					{\\rip}{{$\rip$}}1
					{\\emptyset}{{$\emptyset$}}1
					{\\notin}{{$\notin\ $}}1
					{\\not\\exists}{{$\not\exists\ $}}1
					{\\ne}{{$\ne\ $}}1
					{\\to}{{$\to\ $}}1
					{\\implies}{{$\implies\ $}}1
					{<}{{$<\ $}}1
					{>}{{$>\ $}}1
					{=}{{$=\ $}}1
					{~}{{$\neg\ $}}1
					{|}{{$\mid$}}1
					{'}{{$^\prime$}}1
					{\\A}{{$\forall\ $}}1
					{\\E}{{$\exists\ $}}1
					{\\/}{{$\vee\,$}}1
					{\\vee}{{$\vee\,$}}1
					{/\\}{{$\wedge\,$}}1
					{\\wedge}{{$\wedge\,$}}1
					{->}{{$\rightarrow\ $}}1
					{=>}{{$\Rightarrow\ $}}1
					{->}{{$\rightarrow\ $}}1
					{<=}{{$\Leftarrow\ $}}1
					{<-}{{$\leftarrow\ $}}1
					{~=}{{$\neq\ $}}1
					{\\U}{{$\cup\ $}}1
					{\\I}{{$\cap\ $}}1
					{|-}{{$\vdash\ $}}1
					{-|}{{$\dashv\ $}}1
					{<<}{{$\ll\ $}}2
					{>>}{{$\gg\ $}}2
					{|}{{$\|$}}1
					{[}{{$[$}}1
					{]}{{$\,]$}}1
					{[[}{{$\langle$}}1
					{]]]}{{$]\rangle$}}1
					{]]}{{$\rangle$}}1
					{<=>}{{$\Leftrightarrow\ $}}2
					{<->}{{$\leftrightarrow\ $}}2
					{(+)}{{$\oplus\ $}}1
					{(-)}{{$\ominus\ $}}1
					{_i}{{$_{i}$}}1
					{_j}{{$_{j}$}}1
					{_{i,j}}{{$_{i,j}$}}3
					{_{j,i}}{{$_{j,i}$}}3
					{_0}{{$_0$}}1
					{_1}{{$_1$}}1
					{_2}{{$_2$}}1
					{_n}{{$_n$}}1
					{_p}{{$_p$}}1
					{_k}{{$_n$}}1
					{-}{{$\ms{-}$}}1
					{@}{{}}0
					{\\delta}{{$\delta$}}1
					{\\R}{{$\R$}}1
					{\\Rplus}{{$\Rplus$}}1
					{\\N}{{$\N$}}1
					{\\times}{{$\times\ $}}1
					{\\tau}{{$\tau$}}1
					{\\alpha}{{$\alpha$}}1
					{\\beta}{{$\beta$}}1
					{\\gamma}{{$\gamma$}}1
					{\\ell}{{$\ell\ $}}1
					{--}{{$-\ $}}1
					{\\TT}{{\hspace{1.5em}}}3        
				}

				\lstdefinelanguage{BigPVS}{
					basicstyle=\tt,
					keywordstyle=\sc,
					identifierstyle=\it,
					emphstyle=\tt ,
					mathescape=true,
					tabsize=20,
					sensitive=false,
					columns=fullflexible,
					keepspaces=false,
					flexiblecolumns=true,
					basewidth=0.05em,
					moredelim=[il][\rm]{//},
					moredelim=[is][\sf \figuresize]{!}{!},
					moredelim=[is][\bf \figuresize]{*}{*},
					keywords={and, 
						begin,
						cases, const,
						do, datatype,
						external, else, exists, end, endif, endcases,
						fi,for, forall, from,
						hidden,
						in, if, importing,
						let, lambda, lemma,
						measure,
						not,
						or, of,
						return, recursive,
						stop, 
						theory, that,then, type, types, type+, to, theorem,
						var,
						with,while},
					emph={nat, setof, sequence, eq, tuple, map, array, first, rest, add, enumeration, bool, real, posreal, nnreal},   
					literate=
					{(}{{$($}}1
					{)}{{$)$}}1
					{\\in}{{$\in\ $}}1
					{\\mapsto}{{$\rightarrow\ $}}1
					{\\preceq}{{$\preceq\ $}}1
					{\\subset}{{$\subset\ $}}1
					{\\subseteq}{{$\subseteq\ $}}1
					{\\supset}{{$\supset\ $}}1
					{\\supseteq}{{$\supseteq\ $}}1
					{\\forall}{{$\forall$}}1
					{\\le}{{$\le\ $}}1
					{\\ge}{{$\ge\ $}}1
					{\\gets}{{$\gets\ $}}1
					{\\cup}{{$\cup\ $}}1
					{\\cap}{{$\cap\ $}}1
					{\\langle}{{$\langle$}}1
					{\\rangle}{{$\rangle$}}1
					{\\exists}{{$\exists\ $}}1
					{\\bot}{{$\bot$}}1
					{\\rip}{{$\rip$}}1
					{\\emptyset}{{$\emptyset$}}1
					{\\notin}{{$\notin\ $}}1
					{\\not\\exists}{{$\not\exists\ $}}1
					{\\ne}{{$\ne\ $}}1
					{\\to}{{$\to\ $}}1
					{\\implies}{{$\implies\ $}}1
					{<}{{$<\ $}}1
					{>}{{$>\ $}}1
					{=}{{$=\ $}}1
					{~}{{$\neg\ $}}1
					{|}{{$\mid$}}1
					{'}{{$^\prime$}}1
					{\\A}{{$\forall\ $}}1
					{\\E}{{$\exists\ $}}1
					{\\/}{{$\vee\,$}}1
					{\\vee}{{$\vee\,$}}1
					{/\\}{{$\wedge\,$}}1
					{\\wedge}{{$\wedge\,$}}1
					{->}{{$\rightarrow\ $}}1
					{=>}{{$\Rightarrow\ $}}1
					{->}{{$\rightarrow\ $}}1
					{<=}{{$\Leftarrow\ $}}1
					{<-}{{$\leftarrow\ $}}1
					{~=}{{$\neq\ $}}1
					{\\U}{{$\cup\ $}}1
					{\\I}{{$\cap\ $}}1
					{|-}{{$\vdash\ $}}1
					{-|}{{$\dashv\ $}}1
					{<<}{{$\ll\ $}}2
					{>>}{{$\gg\ $}}2
					{|}{{$\|$}}1
					{[}{{$[$}}1
					{]}{{$\,]$}}1
					{[[}{{$\langle$}}1
					{]]]}{{$]\rangle$}}1
					{]]}{{$\rangle$}}1
					{<=>}{{$\Leftrightarrow\ $}}2
					{<->}{{$\leftrightarrow\ $}}2
					{(+)}{{$\oplus\ $}}1
					{(-)}{{$\ominus\ $}}1
					{_i}{{$_{i}$}}1
					{_j}{{$_{j}$}}1
					{_{i,j}}{{$_{i,j}$}}3
					{_{j,i}}{{$_{j,i}$}}3
					{_0}{{$_0$}}1
					{_1}{{$_1$}}1
					{_2}{{$_2$}}1
					{_n}{{$_n$}}1
					{_p}{{$_p$}}1
					{_k}{{$_n$}}1
					{-}{{$\ms{-}$}}1
					{@}{{}}0
					{\\delta}{{$\delta$}}1
					{\\R}{{$\R$}}1
					{\\Rplus}{{$\Rplus$}}1
					{\\N}{{$\N$}}1
					{\\times}{{$\times\ $}}1
					{\\tau}{{$\tau$}}1
					{\\alpha}{{$\alpha$}}1
					{\\beta}{{$\beta$}}1
					{\\gamma}{{$\gamma$}}1
					{\\ell}{{$\ell\ $}}1
					{--}{{$-\ $}}1
					{\\TT}{{\hspace{1.5em}}}3        
				}
				
				\lstdefinelanguage{pvsNums}[]{pvs}
				{
					numbers=left,
					numberstyle=\tiny,
					stepnumber=2,
					numbersep=4pt
				}
				
				\lstdefinelanguage{pvsNumsRight}[]{pvs}
				{
					numbers=right,
					numberstyle=\tiny,
					stepnumber=2,
					numbersep=4pt
				}

				\lstnewenvironment{BigPVS}%
				{\lstset{language=BigPVS}}
				{}
				
				\lstnewenvironment{PVSNums}%
				{
					\if@firstcolumn
					\lstset{language=pvs, numbers=left, firstnumber=auto}
					\else
					\lstset{language=pvs, numbers=right, firstnumber=auto}
					\fi
				}
				{}
				
				\lstnewenvironment{PVSNumsRight}%
				{
					\lstset{language=pvs, numbers=right, firstnumber=auto}
				}
				{}

				\newcommand{\linefigpvs}[9]{
					
				}

				\lstdefinelanguage{pvsproof}{
					basicstyle=\tt \figuresize,
					mathescape=true,
					tabsize=4,
					sensitive=false,
					columns=fullflexible,
					keepspaces=false,
					flexiblecolumns=true,
					basewidth=0.05em,
				}

				\lstdefinelanguage{pseudo}{
					basicstyle=\figuresize,
					keywordstyle=\bf \figuresize,
					identifierstyle=\it \figuresize,
					emphstyle=\tt \figuresize,
					mathescape=true,
					tabsize=20,
					sensitive=false,
					columns=fullflexible,
					keepspaces=false,
					flexiblecolumns=true,
					basewidth=0.05em,
					moredelim=[il][\rm]{//},
					moredelim=[is][\sf \figuresize]{!}{!},
					moredelim=[is][\bf \figuresize]{*}{*},
					keywords={automaton,and, 
						choose,const,continue, components,
						discrete, do,
						eff, external,else, elseif, evolve, end,
						fi,for, forward, from,
						hidden,
						in,input,internal,if,invariant, initially, imports,
						let,
						or, output, operators, od, of,
						pre,
						return, round,
						such,satisfies, stop, signature, simulation, 
						trajectories,trajdef, transitions, that,then, type, types, to, tasks,
						upon,
						variables, vocabulary, 
						wait, when,where, with,while},
					emph={set, seq, tuple, map, array, enumeration},   
					literate=
					{(}{{$($}}1
					{)}{{$)$}}1
					{\\in}{{$\in\ $}}1
					{\\preceq}{{$\preceq\ $}}1
					{\\subset}{{$\subset\ $}}1
					{\\subseteq}{{$\subseteq\ $}}1
					{\\supset}{{$\supset\ $}}1
					{\\supseteq}{{$\supseteq\ $}}1
					{\\forall}{{$\forall$}}1
					{\\le}{{$\le\ $}}1
					{\\ge}{{$\ge\ $}}1
					{\\gets}{{$\gets\ $}}1
					{\\cup}{{$\cup\ $}}1
					{\\cap}{{$\cap\ $}}1
					{\\langle}{{$\langle$}}1
					{\\rangle}{{$\rangle$}}1
					{\\exists}{{$\exists\ $}}1
					{\\bot}{{$\bot$}}1
					{\\rip}{{$\rip$}}1
					{\\emptyset}{{$\emptyset$}}1
					{\\notin}{{$\notin\ $}}1
					{\\not\\exists}{{$\not\exists\ $}}1
					{\\ne}{{$\ne\ $}}1
					{\\to}{{$\to\ $}}1
					{\\implies}{{$\implies\ $}}1
					{<}{{$<\ $}}1
					{>}{{$>\ $}}1
					{=}{{$=\ $}}1
					{~}{{$\neg\ $}}1
					{|}{{$\mid$}}1
					{'}{{$^\prime$}}1
					{\\A}{{$\forall\ $}}1
					{\\E}{{$\exists\ $}}1
					{\\/}{{$\vee\,$}}1
					{\\vee}{{$\vee\,$}}1
					{/\\}{{$\wedge\,$}}1
					{\\wedge}{{$\wedge\,$}}1
					{=>}{{$\Rightarrow\ $}}1
					{->}{{$\rightarrow\ $}}1
					{<=}{{$\Leftarrow\ $}}1
					{<-}{{$\leftarrow\ $}}1
					{~=}{{$\neq\ $}}1
					{\\U}{{$\cup\ $}}1
					{\\I}{{$\cap\ $}}1
					{|-}{{$\vdash\ $}}1
					{-|}{{$\dashv\ $}}1
					{<<}{{$\ll\ $}}2
					{>>}{{$\gg\ $}}2
					{|}{{$\|$}}1
					{[}{{$[$}}1
					{]}{{$\,]$}}1
					{[[}{{$\langle$}}1
					{]]]}{{$]\rangle$}}1
					{]]}{{$\rangle$}}1
					{<=>}{{$\Leftrightarrow\ $}}2
					{<->}{{$\leftrightarrow\ $}}2
					{(+)}{{$\oplus\ $}}1
					{(-)}{{$\ominus\ $}}1
					{_i}{{$_{i}$}}1
					{_j}{{$_{j}$}}1
					{_{i,j}}{{$_{i,j}$}}3
					{_{j,i}}{{$_{j,i}$}}3
					{_0}{{$_0$}}1
					{_1}{{$_1$}}1
					{_2}{{$_2$}}1
					{_n}{{$_n$}}1
					{_p}{{$_p$}}1
					{_k}{{$_n$}}1
					{-}{{$\ms{-}$}}1
					{@}{{}}0
					{\\delta}{{$\delta$}}1
					{\\R}{{$\R$}}1
					{\\Rplus}{{$\Rplus$}}1
					{\\N}{{$\N$}}1
					{\\times}{{$\times\ $}}1
					{\\tau}{{$\tau$}}1
					{\\alpha}{{$\alpha$}}1
					{\\beta}{{$\beta$}}1
					{\\gamma}{{$\gamma$}}1
					{\\ell}{{$\ell\ $}}1
					{--}{{$-\ $}}1
					{\\TT}{{\hspace{1.5em}}}3
				}

				\newcommand{\reacht}[1]{\relax\ifmmode {\sf Reach}_{#1}(t) \else ${\sf Reach}_{#1}(t)$\fi} 
				
				\newcommand{\reachi}[2]{\relax\ifmmode {\sf Reach}_{#1}(#2) \else ${\sf Reach}_{#1}(#2)$\fi} 
				
				\newcommand{\breach}[2]{\relax\ifmmode {\sf Reach}^{#2}_{#1} \else ${\sf Reach}^{#2}_{#1}$\fi} 
				
				\newcommand{\breachi}[3]{\relax\ifmmode {\sf Reach}^{#2}_{#1}(#3) \else ${\sf Reach}^{#2}_{#1}(#3)$\fi} 
				
				\newcommand{\breacht}[2]{\relax\ifmmode {\sf Reach}^{#2}_{#1}(t) \else ${\sf Reach}^{#2}_{#1}(t)$\fi}

\begin{document}	

\makeatletter
\g@addto@macro \normalsize {%
 \setlength\abovedisplayskip{1pt}%
 \setlength\belowdisplayskip{1pt}%
}
\makeatother

\setlength{\textfloatsep}{1pt}  %space after floats reduction
\renewcommand{\baselinestretch}{1.0}

\title{Approximate Partial Order Reduction\thanks{This work is supported by the grants CAREER 1054247 and 
CCF 1422798 from the National Science Foundation}}
%%
%\titlerunning{Approximate Partial Order Reduction}  % abbreviated title (for running head)
\author{Chuchu Fan, Zhenqi Huang, and Sayan Mitra}
%\institute{University of Illinois at Urbana-Champaign}
%CCC \and DDD \and EEE \and FFF \and
%GGG}
%
%\authorrunning{AAA BBB et al.} % abbreviated author list (for running head)
%
%%% list of authors for the TOC (use if author list has to be modified)
%Craig Chambers, Kim B. Bruce, and Elisa Bertino}
%
\institute{University of Illinois at Urbana-Champaign, ECE Department\\
\email{\{cfan10, zhuang25, mitras\}@illinois.edu}}
%,\\
%\and
%Department of Electrical and Computer Engineering, University of Illinois at Urbana
%Champaign,\\
%\email{BBB@illinois.edu}}

\maketitle              % typeset the title of the contribution

\begin{abstract}
We present a new partial order reduction method for reachability analysis of nondeterministic labeled transition systems over metric spaces. Nondeterminism arises from both the choice of the initial state and the choice of  actions, and the number of executions to be explored grows exponentially with their length. 
We introduce a notion of $\varepsilon$-independence relation over actions that relates approximately commutative actions; $\varepsilon$-equivalent action sequences are obtained by  swapping $\varepsilon$-independent consecutive action pairs.
Our reachability algorithm generalizes individual executions to cover sets of  executions that start from  different, but $\delta$-close initial states, and 
follow different, but $\varepsilon$-independent, action sequences.
The constructed over-approximations can be made arbitrarily precise by reducing the $\delta,\varepsilon$ parameters. 
Exploiting both the continuity of actions and their approximate independence, the algorithm can yield an exponential reduction in the number of executions  explored. We illustrate this with experiments on consensus, platooning, and distributed control examples.
% by a order of $O(n!)$.
%In some models of concurrent systems, this can provide an exponential saving in the number of executions explored for reachability analysis. 
\end{abstract}
%

%% !TEX root = main.tex
%% !TEX root = main.tex
\section{Introduction}

Actions of different computing nodes interleave arbitrarily in distributed systems. The number of action sequences  that have to be examined for state-space exploration grows exponentially with the number of nodes. 
%be extremely large.
%
Partial order reduction methods tackle this combinatorial explosion by eliminating executions that are {\em equivalent\/}, i.e., do not provide new information about reachable states ~(see~\cite{godefroid1996partial,peled1998ten,kurshan1998static} and the references therein). 
%Two executions are deemed equivalent if one of them captures all the relevant information provided by the other.
This equivalence is based on {\em independence\/} of actions: a pair of actions are independent if they commute, i.e., applying them in any order results in the same state. 
Thus, of all execution branches that start and end at the same state,  but perform commuting actions in different order, only one has to be explored.
%
%Equivalent executions are obtained by permuting independent actions as these reorderings  do not provide any additional information, and therefore, can be ignored.
%%
Partial order reduction methods have become standard tools for practical software verification. They have been successfully applied to election protocols~\cite{alur1997partial}, indexers~\cite{flanagan2005dynamic}, file systems~\cite{clarke1999state}, security protocol~\cite{clarke2000partial}, distributed schedulers~\cite{baier2004partial}, among many others.
%In~\cite{clarke1999state}, a file transfer protocol with millions of states and transitions is verified in a matter of seconds using a partial order reduction method. 

Current partial order methods are limited when it comes to computation with numerical data and physical quantities (e.g., sensor networks, vehicle platoons, IoT applications, and distributed control and monitoring systems).
First, a pair of actions are considered independent only if they commute exactly; actions that nearly commute---as are common in these applications---cannot be exploited for pruning the exploration. %Such actions are common in numerically-based programs.  
%
%In the above type of applications, Such exactly independent action pairs are rare in CPS context, where the individual nodes often interact with a shared environment in the presence of noise and disturbances. 
%This makes the partial order reduction methods ignore action pairs which lead to {\em nearly} but not exactly identical states.
Second, conventional partial order methods do not eliminate executions that start from nearly similar states and experience equivalent action sequences. 
%
%For verifying infinite state systems, it is essential to cover a compact set of neighboring states from the exploration of a single starting state. 
%This limitation is serious for verifying infinite state systems as the number of executions can be infinite.

We address these  limitations and propose a state space exploration method for nondeterministic, infinite state transition systems based on {\em approximate partial order reduction\/}. 
%fall in this category 
%
Our setup has two mild assumptions: (i) the state space of the transition system has a discrete part $L$ and a continuous part $X$ and the latter is equipped with a metric; (ii) the actions on  $X$ are continuous functions.
%
%First of all, we assume that there is a metric on the state space.  
Nondeterminism arises from both the choice of the initial state and the choice of actions. 
Fixing an initial state $q_0$ and a sequence of actions $\tau$ (also called a {\em trace}), uniquely defines an execution of the system which we denote by $\xi_{q_0, \tau}$.
For a given approximation parameter $\varepsilon\geq 0$, we define two actions $a$ and $b$ to be $\varepsilon$-independent if from any state $q$, the continuous parts of states resulting from applying action sequences $ab$ and $ba$ are  $\varepsilon$-close.
Two {\em traces} of $\A$ are  $\varepsilon$-equivalent if they result from permuting $\varepsilon$-independent actions.
To compute the reachable states of $\A$ using a finite (small) number of executions, the key is to generalize or expand an execution $\xi_{q_0, \tau}$ by a factor $r \geq 0$, so that, this expanded set contains all executions that start $\delta$-close to $q_0$ and experience action sequences that are $\varepsilon$-equivalent to $\tau$. We call this $r$ a {\em $(\delta,\varepsilon)$-trace equivalent discrepancy factor ($\bloatf$)} for $\xi$.

For a fixed  trace $\tau$, the only source of nondeterminism is the choice of the initial state. The reachable states from $B_\delta(q_0)$---a $\delta$-ball around $q_0$---can be over-approximated by expanding $\xi_{q_0, \tau}$ by a  $(\delta,0)$-$\bloatf$. This is essentially the sensitivity of $\xi_{q_0, \tau}$ to $q_0$. Techniques for computing it are now well-developed for a broad class of models ~\cite{DMVemsoft2013,breachAD,c2e2tacas,FanMitra2015}.

Fixing $q_0$, the only source of nondeterminism is the possible sequence of actions in $\tau$.
The reachable states from $q_0$ following all possible valid traces can be over-approximated by expanding $\xi_{q_0,\tau}$ by a $(0,\varepsilon)$-$\bloatf$, which includes states reachable by all $\varepsilon$-equivalent action sequences. Computing  $(0,\varepsilon)$-$\bloatf$ uses the principles of partial order reduction. However, unlike exact equivalence, here,  starting from the same state, the states reached at the end of executing two $\varepsilon$-equivalent traces are not necessarily identical. 
This breaks a key assumption necessary for conventional partial order algorithms: here, an action enabled after $ab$ may not be enabled after $ba$. Of course, considering disabled actions can still give  over-approximation of reachable states, but, we show that the precision of approximation can be improved arbitrarily by shrinking $\delta$ and $\varepsilon$.

Thus, the reachability analysis in this paper brings together two different ideas for handling nondeterminism:  it combines sensitivity analysis with respect to initial state and $\varepsilon$-independence of actions in computing  $(\delta,\varepsilon)$-$\bloatf$, i.e., 
 upper-bounds on the distance between executions starting from initial states that are $\delta$-close to each other and follow $\varepsilon$-equivalent action sequences (Theorem~\ref{them:soundness}). 
As a matter of theoretical interest, we  show that the approximation error can be made arbitrarily small by choosing sufficiently small $\delta$ and $\varepsilon$  (Theorem~\ref{them:precision}).
%
%With this extension, the reduction can be applied to systems where the conventional methods may not.
We validate the correctness and effectiveness of the algorithm with three case studies where conventional partial order reduction would not help: an iterative consensus protocol, a simple vehicle platoon control system, and a distributed building heating system. In most cases, our reachability  algorithm reduces the number of explored executions by a factor of $O(n!)$, for a time horizon of $n$, compared with exhaustive enumeration. 
%While the computation for $(\delta,\varepsilon)$-$\bloatf$ only takes at most $O(n^2)$ time at each step. 
Using these over-approximations, we could quickly decide safety verification questions.
%
%over-approximates the reachset state from a small number of representative executions and conventional partial order reduction methods do not apply. 
These examples illustrate that our method has the potential to improve verification of  a broader range of distributed systems for consensus~\cite{blondel2005convergence,fang2005information,olfati2007consensus,mitra2008formalized}, synchronization~\cite{welch1988new,rhee2009clock} and control~\cite{fehnker2004benchmarks,mitra1994asynchronous}. 

%%% Related work

\paragraph{Related work.}
There are two main classes of partial order reduction methods. 
The {\em persistent/ample set} methods compute a subset of enabled transitions --the persistent set (or ample set)-- such that the omitted transitions are independent to those selected~\cite{clarke1999model,alur1997partial}.
The reduced system which only considers the transitions in the persistent set is guaranteed to represent all behaviors of the original system.
The persistent sets and the reduced systems are often derived by static analysis of the code.
More recently, researchers have developed dynamic partial order reduction methods using the {\em sleep set} to avoid the static analysis~\cite{yang2008efficient,abdulla2014optimal,flanagan2005dynamic}.
These methods examine the history of actions taken by an execution and decide a set of actions that need to be explored in the future.
The set of omitted actions is the sleep set.
In \cite{cassez2015verification}, \chuchu{Cassez and Ziegler} introduce a method to apply symbolic partial order reduction to infinite state discrete systems.
%The persistent set and sleep set methods can be used complementary~\cite{flanagan2005dynamic}.

Analysis of sensitivity and the related notion of robustness analysis  functions, automata, and executions has recently received significant attention~\cite{chaudhuri2012continuity,breachAD,samanta2013robustness}.
\chuchu{Majumdar and Saha}~\cite{majumdar2009symbolic} present an algorithm to compute the output deviation with bounded disturbance combining symbolic execution and optimization.
In~\cite{chaudhuri2012continuity} and~\cite{samanta2013robustness}, \chuchu{Chaudhuri etc.,} present algorithms for robustness analysis of programs and  networked systems.
Automatic techniques for local sensitivity analysis combining simulations and static analysis and their applications to verification of hybrid systems have been presented in~\cite{breachAD,c2e2tacas,FanMitra2015}.

In this paper, instead of conducting conventional partial order reduction, we propose a novel method of approximate partial order reduction, and combine it with sensitivity analysis for reachability analysis and safety verification for a broader class of systems. 
%A previously unpublished manuscript with some of the core ideas of this paper was posted online at~\cite{HM:TACAS-full}.

\section{Preliminaries}
\label{por:sec:infts}

\paragraph{Notations.}
The state of our labeled transition system is defined by the valuations of a set of variables. Each variable $v$ has a type, $\mathit{type}(v)$, which is either the set of reals or some finite set. For a set of variables $V$, a valuation $\vv$ maps each $v \in V$ to a point in $\mathit{type}(v)$. The set of all valuations of $V$ is $\val{V}$.
$\reals$  denotes the set of reals, $\nnreals$ the set of non-negative reals, 
and $\naturals$ the set of natural numbers. 
For $n\in \naturals$, $[n] = \{0,\ldots, n-1\}$.
The spectral radius $\rho(A)$ of a square matrix $A \in \reals^{n \times n}$ is the largest absolute value of its eigenvalues. A square matrix $A$ is {\em stable} if its spectral radius $\rho(A) <1$.
For a set of tuples $S  = \{\langle s_{j1},\dots,s_{jn} \rangle_j\}$, $S \restr i$ denotes the set $\{s_{ji}\}$ which is the set obtained by taking the $i^{th}$ component of each tuple in $S$. 

\subsection{Transition systems}
\label{por:sec:model}
\begin{definition}
A {\em labeled transition system\/} $\A$ is a tuple $\langle X\cup L, \Theta, A, \arrow{} \rangle$ where
\begin{inparaenum}[(i)]
\item $X$ is a set of {\em real-valued variables} and $L$ is a set of {\em finite-valued variables}.
$Q = \val {X\cup L}$ is the {\em set of states}, 
\item \chuchu{$\Theta \subseteq Q$ is a set of {\em initial states} such that the sets of real-valued variables are compact},
\item $A$ is a finite set of {\em actions}, and
\item $\rightarrow \subseteq Q \times A \times Q$ is a {\em transition relation}.
% specifically, the transition relation is specified by a set of transition functions.
%\item 
\end{inparaenum}
\end{definition}
%%%
%%%
A state $q\in Q$ is a valuation of  the real-valued and finite-valued variables. We denote by $q.X$ and $q.L$, respectively, the real-valued and discrete (finite-valued) parts of the state $q$. 
We will view the continuous part $q.X$  as a vector in $\reals^{|X|}$ by fixing an arbitrary ordering of $X$.
The norm $|\cdot|$ on $q.X$ is an arbitrary norm unless stated otherwise.
% For implementations later in the paper we may have to spell it out.
%
For $\delta \geq 0$, the {\em $\delta$-neighborhood} of $q$ is denoted by $\B_{\delta}(q) \deq \{q'\in Q :  q'.L=q.L \wedge |q'.X-q.X| \leq \delta\}$.
For any $(q,a,q')\in\arrow{}$, we write $q\arrow a q'$.
%For any state $q\in Q$, we define $enabled(q) = \{a\in A : \exists q'. q \arrow a q'\}$ as the set of actions {\em enabled} at state $q$. 
For any action $a\in A$, its guard is the set  $\mathit{guard}(a) = \set{q\in Q}{\exists q'\in Q, q \arrow a q'}$. 
%\chuchu{
	We assume that guards are closed sets.
%	 $\{q.X \ | \ q \in \mathit{guard}(a) \}$ is a closed set.}
%
An action $a$ is {\em deterministic} if for any state $q\in Q$, if there exists $q_1, q_2\in Q$ with $q \arrow a q_1$ and $q \arrow a q_2$, then $q_1 = q_2$.
\begin{assumption}
\label{as:det}
\begin{inparaenum}[(i)]
\item Actions are deterministic.
% and specified by a guard and a {\em transition function}.
For notational convenience, the name of an action $a$ is identified with  its transition function, i.e.,  for each $q\in \mathit{guard}(a)$, $q \arrow a a(q)$. 
We extend this notation to all states, i.e., even those outside $\mathit{guard}(a)$.
\item  For any state pair $q, q'$, if $q.L=q'.L$ then   $a(q).L = a(q').L$.
\end{inparaenum}
\end{assumption}

%
%\begin{enumerate}[(i)]
%\item for any $q\in \mathit{guard}(a)$, $(q, a,\va(q))\in \rightarrow$,
%\item for any $q,q'\in enabled(a)$, if $q.L=q'.L$ then $\va(q).L = \va(q').L$, and
%%\item there exists $L\geq 0$ such that for any $q,q'\in enabled(a)$ with $q.L = q.L'$, $|\va(q).X-\va(q').X| \leq L|q.X-q'.X|$ that is, $\va$ is Lipschitz in $\val X$.
%\end{enumerate}
%\end{assumption}

%It is straightforward to see that any deterministic transitions system with no real-valued variables ($X = \emptyset$) satisfies assumption~\ref{por:ass:main}, where the transition function $a(q)$ is defined as $q'$ if $q \arrow a q'$ and otherwise as $q$.

%we define the transition function 
%$a(q)$ as $q'$ if $(q,a,q')\in \rightarrow$ and otherwise as $q$. 

%We assume that all the actions in $A$ to be deterministic
%and the resulting transition system called a deterministic labeled transition system. 
%For a deterministic action $a$, 
%we define the transition function 
%$a(q)$ as $q'$ if $(q,a,q')\in \rightarrow$ and otherwise as $q$. 

\paragraph{Executions and traces.}
%Let $A^*$ denote the set of finite {\em action sequences}.
For a deterministic transition system, a state $q_0\in Q$ and a finite action sequence (also called a {\em trace}) $\tau = a_0a_1\dots a_{n-1}$ uniquely specifies a {\em potential execution\/} $\xi_{q_0,\tau} = q_0, a_0, q_1, a_1, \dots, a_{n-1}, q_n$ where for each $i\in[n],$ 
%if $a_i$ is enabled at state $q_i$, then 
$a_i(q_i) = q_{i+1}$. 
%\sayan{(1) We need to first define what is an execution/potential execution, before making this statement. 
%(2) Also, in going from potential execution to execution, do we not also need to say that the reset function gives the right state (in addition to the guards being enables) ?}
%\textcolor{blue}{why potential? Isn't it the definition conflict with valid execution and also conflict with the later part of Figure 2?} $\xi_{q_0,\tau} = q_0, a_0, q_1, a_1, \dots, a_{n-1}, q_n$ where for each $i\in[n],$ $a_i(q_i) = q_{i+1}$. 
%We denote such a potential execution by $\exe{q_0, \tau}$.
%\chuchu{
	A {\em valid execution\/} (also called execution for brevity) is a potential execution with 
(i) $q_0 \in \Theta$ and 
(ii) for each $i\in[n]$, $q_i\in \mathit{guard}(a_i)$.  That is, a valid execution is a potential execution starting from the initial set with each action $a_i$ enabled at state $q_i$.
%}
For any potential execution $\xi_{q_0,\tau}$, its {\em trace\/} is the  action sequence $\tau$, i.e., $\trace{\exe{q_0,\tau}}= \tau \in A^*$.
We denote by $len(\tau)$ the length of $\tau$. For any for $i\in [len(\tau)]$, $\tau(i)$ is the $i$-{th} action in $\tau$.
The length  of $\xi_{q_0,\tau}$ is the length of its trace and $\xi_{q_0,\tau}(i) = q_i$ is the state visited after the $i$-th transition. 
The first and last state on a execution $\xi$ are denoted as $\xi.\fstate = \xi(0)$ and $\xi.\lstate = \xi(len(\xi))$.
%\item $Trace(\xi, i) = \tau(i) = a_i$ is the transition taken at step $i$.
%For $i\in [len(\xi)]$, $\mathit{prefix}_i(\xi)$ and $\mathit{suffix}_i(\tau)$ are the length-$i$ prefix and suffix of $\xi$. 

For a subset of initial states $S\subseteq \Theta$ and a time bound $T\geq 0$, $\execs{S,T}$ is the set of length $T$ executions starting from  $S$. 
%$\traces{S,T} \deq \set{\trace{\xi}}{\xi \in \execs{S,T}}$ is the set of length $T$ traces corresponding to executions from $S$.
We denote the {\em reach set at time $T$} by $\reach{}{S,T} \deq \set{\xi.\lstate}{\xi \in \execs{S,T}}$.
Our goal is to precisely over-approximate  $\reach{}{\Theta,T}$ exploiting  partial order reduction.
%\sayan{I wonder if it makes more sense to change the problem statement to bounded safety verification.}
\vspace{-0.2in}
\begin{figure}[bhtp]
\centering
  \hrule
  \two{.47}{.47}
  {
\begin{lstlisting}[language=ioaNums,lastline=7]
automaton $\auto{Consensus}(n \in \naturals, N \in \naturals)$
	variables
		$x : \reals^n$
		$d :\booleans^N$
	initially
		$x[i]\in[-4,4]$ for each $i\in [n]$
		$d[i] := \mathit{false}$ for each $i\in[n]$
\end{lstlisting}}
  {
\begin{lstlisting}[language=ioaNumsRight,firstline=8]
	transitions
		$\act{a_i}$ for each $i\in[N]$ 
			pre 	$\neg d_i$
			eff 	$x := A_i x \ \wedge \ d[i] := \true$
		$\act{a_\bot}$
			pre   $\wedge_{i\in[N]} d[i]$
			eff	$d[i] :=\mathit{false}$ for each $i\in[N]$
\end{lstlisting}}
  \hrule
  \caption{\small Labeled transition system model of iterative consensus.}
  \label{por:hioa:consensus}
\end{figure}
\vspace{-0.25in}
%\vspace{-1cm}

\begin{example}[Iterative consensus]
\label{por:ex:linear}
An $n$-dimensional iterative consensus protocol with $N$ processes is shown in Figure~\ref{por:hioa:consensus}. 
The real-valued part of state is a vector $x$ in $\reals^n$
and each process $i$ changes the state by the linear transformation $x \gets A_i x$. 
The system evolves in rounds: in each round, each process $i$ updates the state exactly once but in arbitrary order.
%---this gives fairness and can be generalized. 
The boolean vector $d$ marks the processes that have acted in a round.
\chuchu{The set of actions is $\{a_i\}_{i\in[N]} \cup \{a_\bot\}$.}
%The guard and transition function of each action are defined by the  precondition ({\bf pre}) and effect ({\bf eff}) statements.
For each $i \in [N]$, the action $a_i$ is enabled when $d[i]$ is $\mathit{false}$ and 
when it occurs $x$ is updated as $A_i x$, where $A_i$ is an $n\times n$ matrix.
The action $a_\bot$ can occur only when all $d[i]$'s are set to $\mathit{true}$ and it 
resets all the $d[i]$'s to $\mathit{false}$. 
%is 
%The boolean $d_i$'s are used to ensure that all the actions $a_i$ occur before any can be repeated. 
%If for all $i\in[N]$ $d_i =\mathit{\true}$, then the action $a_\bot$ is enabled and it resets $d_i$ to $\mathit{\false}$.
For an instance with $N=3$, a valid execution could have the trace $\tau = a_0a_2a_1a_\bot a_1a_0a_2a_\bot$. 
It can be checked that Assumption~\ref{as:det} holds.
% Each action is  deterministic. For any $q,q'\in Q$ and $a_i\in A$, 
%if $q.d = q'.d$, action $a_i$ simply sets $d_i$ to $\true$ and keeps all other components unchanged,  
%and hence, $a_i(q).d = a_i(q').d$. 
%Action $a_\bot$ resets all $d_i$, hence $a_\bot(q).d = a_\bot(q').d$.
%Hence, this transition system satisfies Assumption~\ref{por:ass:main}.
In fact, the assumption will continue to hold if $A_i x$ is replaced by a nonlinear transition function $a_i:\reals^n\rightarrow \reals^n$.

%where the reset action $a_\bot$ is applied exactly once after all of the other actions $\{a_0,a_1,a_2\}$ are applied once.
%
%Thus, the linear transformations can be applied in any order and one application of each constitutes a ``round''. This  can be viewed as an abstraction of iterative  consensus~\cite{fang2005information,dpconsensus}, where  $N$ processes perform computations on a shared state.  
\end{example}

\subsection{Discrepancy functions}
\label{sec:discrepancy}
A discrepancy function bounds the changes in a system's executions as a continuous  function of the changes in its inputs.  
Methods for computing discrepancy of dynamical and hybrid systems are now well-developed~\cite{huang2014proofs,c2e2tacas,donze2007systematic}. We extend the notion naturally to labeled transition systems: a discrepancy for an action bounds the changes in the continuous state brought about by its transition function. 

% we use here quantify the continuity property for the transition functions of the system.

%As we discussed in Chapter~\ref{ch:isd} and the authors show in~\cite{}, for computing reach set from an uncountable set of initial states, 
%one can exploit the sensitivity of the executions on initial states.
\begin{definition}
\label{por:def:discrep}
For an action $a\in A$, a continuous function $\beta_a : \nnreals \rightarrow \nnreals$ is a {\em discrepancy function} if for any pair of states $q,q'\in Q$ with $q.L=q'.L$,
\begin{inparaenum}[(i)]
\item $|a(q).X - a(q').X| \leq \beta_a(|q.X - q'.X|)$, and
\item $\beta_a(\cdot) \rightarrow 0$ as $|q.X - q'.X| \rightarrow 0$.
\end{inparaenum}
\end{definition}
Property~(i) gives an upper-bound on the changes brought about by action $a$ and (ii) ensures that the bound given by $\beta_a$ can be made arbitrarily precise. 
If the action $a$ is Lipschitz continuous with Lipschitz constant $L_a$, then $\beta_a(|q.X - q'.X|) = L_a (|q.X - q'.X|)$ can be used as a discrepancy function. Note that we do not assume the system is stable.
%
%Roughly, discrepancy functions captures that executing action $a$ from two states that are close to each other should result in states that remain close.
%We note that the discrepancy function $\beta_a$  is a global property of the transition function $a$,
%which holds over the state space $Q$.
As the following proposition states, 
given discrepancy functions for actions, we can reason about distance between executions that share the same trace but have different initial states. 
%using their initial distance.
\begin{proposition}
\label{por:prop:discrep}
Suppose each action $a \in A$ has a discrepancy function $\beta_a$. 
For any $T\geq 0$ and action sequence $\tau  = a_0a_1a_2\dots a_{T}$, and for any pair of states $q,q'\in Q$ with $q.L=q'.L$, 
the last states of the pair of potential executions satisfy:
% $\xi = \exe{q,\tau}$ and $\xi' = \exe{q',\tau}$ satisfy
\begin{align}
\xi_{q,\tau}.\lstate.L &= \xi_{q',\tau}.\lstate.L, \\
|\xi_{q,\tau}.\lstate.X - \xi_{q',\tau}.\lstate.X| & \leq \beta_{a_T}\beta_{a_{T-1}}\dots\beta_{a_0}(|q.X-q'.X|).
\end{align}
\end{proposition}
%This proposition allows us to compute the distance between a pair of potential executions from their initial distance.
%We will derive the discrepancy functions for the linear transition systems in the following example.
%Later in Section~\ref{por:sec:reach} we 

\begin{example}
\label{por:ex:linear1}
Consider an instance of $\auto{Consensus}$ of Example~\ref{por:ex:linear} with $n = 3$ and $N = 3$ with the standard $2$-norm on $\reals^3$.
Let the matrices $A_i$ be
\[
\scriptsize
A_0 =
\begin{bmatrix}
0.2  &-0.2&  -0.3\\
  -0.2  & 0.2 & -0.1\\
  -0.3  &-0.1 &  0.3
\end{bmatrix}, 
A_1 = 
\begin{bmatrix}
0.2  & 0.3  & 0.2\\
   0.3 & -0.2  & 0.3\\
   0.2  & 0.3  & 0
\end{bmatrix}, 
A_2 =
\begin{bmatrix}
-0.1  & 0  & 0.4\\
   0 & 0.4 & -0.2\\
   0.4&  -0.2 & -0.1
\end{bmatrix}.
\]
\normalsize
%\chuchu{The spectral radius $\rho(A)$ of a square matrix $A \in \reals^{n \times n}$ is defined as the largest absolute value of its eigenvalues. A square matrix $A$ is called {\em stable} if its spectral radius $\rho(A) <1$.
%All the matrices are stable, and}
%\textcolor{red}{Chuchu: I deleted the stable matrix here since they are not necessary for these discrepancy functions and made the reviewers confused.}
It can be checked that for any  pair $q,q'\in Q$ with $q.L=q'.L$,
$|a_i(q).X - a_i(q').X|_2 \leq |A_i|_2 |q.X-q'.X|_2$.
Where the induced 2-norms of the matrices are $|A_0|_2 = 0.57,|A_1|_2=0.56,|A_2|_2=0.53$.
Thus, for any $v\in\nnreals$, we can use  discrepancy functions for $a_0,a_1,a_2$: $\beta_{a_0}(v) =  0.57v, \beta_{a_1}(v) = 0.56 v, \mbox{ and } \beta_{a_2}(v) = 0.53v$.
\end{example}
\chuchu{For actions with nonlinear transition functions, computing global discrepancy functions is difficult in general but local approaches using the eigenvalues of the Jacobian matrices are adequate for computing reachable sets from compact initial sets~\cite{FanMitra2015,huang2015simulation}.}

\subsection{Combining sets of discrepancy functions}
For a finite set of discrepancy functions $\{\beta_a\}_{a\in A'}$ 
corresponding to a set of actions $A' \subseteq A$, 
we define $\beta_{max} = \max_{a\in A'}\{ \beta_a\}$ as $\beta_{max}(v) = \max_{a \in A'} \beta_a(v)$, for each $v \geq 0$. 
%\sayan{Give the formal definition.}
%
From Definition~\ref{por:def:discrep}, for each $a\in S$, $\beta_a(|q.X-q'.X|)\rightarrow 0$ as $|q.X-q'.X|\rightarrow 0$. Hence, as the maximum of $\beta_a$, we have $\beta_{max}(|q.X-q'.X|)\rightarrow 0$ as $|q.X-q'.X|\rightarrow 0$.
It can be checked that $\beta_{max}$ is a discrepancy function of each $a\in S$.

%
%\sayan{Zhenqi. Do we need to argue about $\beta_{max}$ being a discrepancy function?}
%

For $n\geq 0$ and a function $\beta_{max}$ defined as above, we define a function $\gamma_n =\sum_{i=0}^n \beta_{max}^i$; here $\beta^i = \beta \circ \beta^{i-1}$ for $i \geq1 $ and $\beta^0$ is the identity mapping.
%We note that for any $n\in \naturals$,
%the function $\gamma_n$ is uniquely specified by a set of discrepancy functions $\{ \beta_a\}_{a\in S}$.
%In the lemma above, we introduce a sequence of functions $\{\gamma_i\}_{i\in[n]}$ to upper bound distance between execution fragments.
%The sequence of $\{\gamma_i\}_{i\in[n]}$ functions is specified by a function $\beta:\nnreals\rightarrow\nnreals$, which is the upper bound of a set of discrepancy functions. 
Using the properties of discrepancy functions as in Definition~\ref{por:def:discrep}, we can show the following properties of $\{\gamma_n\}_{n\in \naturals}$.
%Full proof of this and other results can be found in~\cite{HM:TACAS-full}.
% resume here
\begin{proposition}
	\label{por:prop:gamma}
	Fix a finite set of discrepancy functions $\{\beta_a\}_{a\in A'}$ with $A'\subseteq A$.
	Let $\beta_{max} = \max_{a\in A'}\{\beta_a\}$. 
	For any $n\geq 0$, $\gamma_n = \sum_{i =0}^n \beta_{max}^i$ satisfies
	\begin{inparaenum}[(i)]
		\item $\forall \ \varepsilon\in \nnreals$ and any $n\geq n' \geq 0$, $\gamma_{n}(\varepsilon) \geq \gamma_{n'}(\varepsilon)$, and
		\item $\lim_{\varepsilon\rightarrow 0}\gamma_n(\varepsilon) = 0$.
	\end{inparaenum}
\end{proposition}
\begin{proof}
$(i)$ For any $n\geq 1$, we have $\gamma_n - \gamma_{n-1} = \beta_{max}^n$.
Since $\beta_{max}^n = \max_{a\in S}\{\beta_a\}$ for some finite $S$, using Definition~\ref{por:def:discrep},
$\beta_{max}^n$ takes only non-negative values.
Hence, the sequence of functions $\{\gamma_n\}_{n\in \nnreals}$ is non-decreasing. \\
$(ii)$ 
%Recall that $\beta = \max\{\beta_a\}_{a\in S}$ is the maximum of a finite set of discrepancy functions. 
Using the property of discrepancy functions, we have $\lim_{\varepsilon\rightarrow 0}\beta_{max}(\varepsilon) = 0$.
By induction on the nested functions, we have $\lim_{\varepsilon\rightarrow 0}\beta_{max}^i(0)$ for any $i\geq 0$.
Hence for any $n\in\nnreals$, $\lim_{\varepsilon\rightarrow 0}\gamma_n(\varepsilon) = \lim_{\varepsilon\rightarrow 0} \sum_{i=0}^{n} \beta_{max}^i(\varepsilon) = 0$.
\qed
\end{proof}
The function $\gamma_n$ depends on the set of $\{ \beta_a\}_{a \in A'}$, but as the $\beta$s will be fixed and clear from context, we write $\gamma_n$ for brevity.

\section{Independent actions and neighboring executions} 
\label{por:sec:indact}

Central to partial order methods is the notion of independent actions. A pair of actions are independent if from any state, the occurrence of the two actions, in either order, results in the  same state. \chuchu{We extend this notion and define a pair of actions to be  $\varepsilon$-independent (Definition~\ref{por:def:appind}), for some $\varepsilon >0$, if the continuous states resulting from swapped action sequences are within $\varepsilon$ distance.}

\subsection{Approximately independent actions}
%In this section, we define a independence relation on the set of actions $A$.
\begin{definition}
\label{por:def:appind}
For $\varepsilon \geq 0 $, two distinct actions $a, b \in A$ are {\em $\varepsilon$-independent}, denoted by $a\mysim{\varepsilon} b$, if for any state $q\in Q$
\begin{inparaenum}[(i)]
%\item (Enableness) $a(q)\in \mathit{guard}(b)$ and $b(q)\in \mathit{guard}(a)$,
\item (Commutativity) $ab(q).L =  ba(q).L$, and
\item (Closeness) $|ab(q).X - ba(q).X| \leq \varepsilon$.
\end{inparaenum}
\end{definition}
The parameter $\varepsilon$ captures the degree of the approximation.
Smaller the value of $\varepsilon$, more restrictive the independent relation.
If $a$ and $b$ are $\varepsilon$-independent with $\varepsilon=0$, then $ab(q) = ba(q)$
and the actions are independent in the standard sense (see e.g. Definition 8.3 of~\cite{baier2008principles}). 
%
%Sayan: the previous sentence is not correct, given point 2 below... but we will leave it for now.
%
Definition~\ref{por:def:appind} extends the standard definition in two ways. 
First,  $b$ need not be  enabled at state $a(q)$, and vice versa.
That is, if $\exe{q_0,ab}$ is an execution, we can only infer that $\exe{q_0,ba}$ is a  potential execution and not necessarily an execution.
Secondly, with $\varepsilon > 0$,  the continuous states can mismatch by $\varepsilon$ when $\varepsilon$-independent actions are swapped.  Consequently, an action $c$ may be enabled at $ab(q)$ but not at $ba(q)$. If $\exe{q_0,abc}$ is a valid execution, we can only infer that
$\exe{q_0,bac}$ is a  potential execution and not necessarily an execution.
%

%As we discussed in the previous paragraph, swapping $\varepsilon$-independent actions does not maintain the enable-ness property of an execution any way.
%Hence, making this extension does not introduce additional loss of precision.

We assume that the parameter $\varepsilon$ does not depend on the state $q$. When computing the value of $\varepsilon$ for concrete systems, we could first find an invariant for the state's real-valued variable $q.X$ such that $q.X$ is bounded, then find an upper-bound of  $|ab(q).X - ba(q).X|$ as $\varepsilon$. For example, if $a$ and $b$ are both linear mappings with $a(q).X  = A_1 q.X + b_1$ and $b(q).X = A_2 q.X + b_2$ and there is an invariant for $q.X$ is such that $|q.X| \leq r $, then it can be checked that  $|ab(q).X - ba(q).X| = | (A_2A_1 - A_1A_2) q.X + (A_2b_1 - A_1 b_2 + b_2 - b_1)| \leq  | A_2A_1 - A_1A_2|r + |A_2b_1 - A_1 b_2 + b_2 - b_1|$.

For a trace  $\tau \in A^*$ and an action $a\in A$,
$\tau$ is $\varepsilon$-independent to $a$, written as $\tau\mysim\varepsilon a$, if $\tau$ is empty string or for every $i\in[len(\tau)]$, $\tau(i)\mysim{\varepsilon} a$.
It is clear that the approximate independence relation over $A$ is symmetric, but  not necessarily transitive. 
%We give instances of both transitive and non-transitive independence relations in the following example.

% POSSIBLEDROP
\begin{example}
\label{por:ex:linear2}
Consider approximate independence of actions in $\auto{Consensus}$.
Fix any $i,j\in[N]$ such that $i\neq j$ and any state $q\in Q$.
It can be checked that: 
%\[
$a_ia_j(q).d[k] = a_ja_i(q).d[k] = \true$
%\left\{
%\begin{array}{ll}
%\true, & \ 
$\mathit{if} \  k\in\{i,j\}$, otherwise
it is 
%\\
$q.d[k]$.
% &\mathit{otherwise}.
%\end{array}
%\right. 
%\]  
Hence, we have $a_ia_j(q).d=a_ja_i(q).d$ and the commutativity condition of Definition~\ref{por:def:appind} holds.
For the closeness condition, we have
% \begin{equation}
% \eqlabel{por:eq:ex1}
$ |a_ia_j(q).x-a_ja_i(q).x|_2 = |(A_iA_j - A_jA_i)q.x|_2 \leq |A_iA_j-A_jA_i|_2 |q.x|_2.$
% \end{equation}
If the matrices $A_i$ and $A_j$ commute, then $a_i$ and $a_j$ are $\varepsilon$-approximately independent with $\varepsilon = 0$.  

%\sayan{There are some general statements here, we need to emphasize them. Pull some of these out to prelims?}
%If $Q$ is a compact set or the system has a bounded invariant set, then $|q.X|_2$ is bounded and there  exists a finite $\varepsilon \geq 0$ such that $a_i\mysim{\varepsilon} a_j$.
%For any any state $q$ and any $i\in[3]$, 
%  the change in the squared 2-norm is: 
%\[
%$|a_i(q).x|_2^2 - |q.x|_2^2 = q.x^\top A_i^\top A_iq.x -  q.x^\top q.x =  q.x^\top(A_i^\top A_i - I_3) q.x$.
%\]
%The specific matrices presented in Example~\ref{por:ex:linear1} are  stable, so $A_i^\top A_i - I_3$ is a negative definite matrix for each $i\in[3]$.
%Hence $|a_i(q).x|_2 \leq |q.x|_2$.
% holds for each $i\in[3]$ and any state $q$.
%That is, the norm of the continuous state is non-increasing.
Suppose initially $x\in[-4,4]^3$ then the 2-norm of the initial state is bounded by the value $4\sqrt 3$. The specific matrices $A_i, i \in [3]$ presented in Example~\ref{por:ex:linear1} are all stable, so $|a_i(q).x|_2 \leq |q.x|_2$, for each $i \in [3]$ and 
the norm of state is non-increasing in any transitions. Therefore, $Inv=\{x\in\reals^3: |x|_2 \leq 4\sqrt 3\}$ is an invariant of the system. 
%From~(\ref{eq:por:eq:ex1}), 
Together, we have $|a_0a_1(q).x-a_1a_0(q).x|_2 \leq 0.1$, $|a_0a_2(q).x-a_2a_0(q).x|_2 \leq 0.07$, and $|a_1a_2(q).x-a_2a_1(q).x|_2 \leq 0.17$.
%\sayan{Why?}
Thus, with $\varepsilon = 0.1$, it follows that $a_0\mysim\varepsilon a_1$ and 
$a_0\mysim\varepsilon a_2$ and $\mysim\varepsilon$ is not transitive, but with 
$\varepsilon=0.2$, $\mysim\varepsilon$ is transitive.
%\sayan{We did not talk about transitivity until now, its strange to introduce it inside the example. Why? Why is the notion important?}
\end{example}

\subsection{$(\delta,\varepsilon)$-trace equivalent discrepancy for action pairs}
\label{sec:disc:ind}
Definition~\ref{por:def:appind} implies that from a single state $q$, executing two $\varepsilon$-independent actions in either order, we end up in states that are within $\varepsilon$ distance. The following proposition uses discrepancy to bound the distance between states reached after performing $\varepsilon$-independent actions starting from {\em different\/} initial states $q$ and $q'$.
% (\chuchu{the proof is given in the full version~\cite{HM:TACAS-full}}).
%\sayan{Citing this older version, with different set of authors, is ok only if things are consistent. Otherwise, it can lead to confusion. We could just cite very specifically a particular theorem / proof in a particular page that is consistent.
%This is important and novel enough that I'd like to give the proof in the appendix if FM page limits permit.}

\begin{proposition}
	\label{prop:disc:inde:pair}
If a pair of actions $a,b\in A$ are $\varepsilon$-independent, and the two states $q,q'\in Q$ satisfy $q.L = q'.L$,
%$q,q'\in \mathit{guard}(a) \cap \mathit{guard}(b)$,
then we have 
\begin{inparaenum}[(i)]
\item $ba(q).L = ab(q').L$, and
\item
$|ba(q).X - ab(q').X| \leq  \beta_b \circ \beta_a(|q.X-q'.X|) + \varepsilon$, where $\beta_a,\beta_b$ are discrepancy functions of $a,b$ respectively.
\end{inparaenum}
\end{proposition}
\begin{proof}
Fix a pair of states $q,q'\in Q$ with $q.L = q'.L$.
Since $a\mysim \varepsilon b$,  we have $ba(q).L = ab(q).L$. 
Using the Assumption, we have $ab(q).L = ab(q').L$. 
Using triangular inequality, we have
% resume here
$|ba(q).X - ab(q').X|\leq |ba(q).X - ba(q').X| + |ba(q').X - ab(q').X|.$
The first  term is bounded by $\beta_{b} \circ \beta_a(|q.X-q'.X|)$ using Proposition \ref{por:prop:discrep} and 
the second is bounded by $\varepsilon$ by Definition \ref{por:def:appind}, and hence, the result follows.
%Therefore $|ba(q).X - ab(q').X| \leq \beta_b\beta_a(|q.X-q'.X|) + \varepsilon$.
\qed
\end{proof}

% !TEX root = main.tex

\section{Effect of $\varepsilon$-independent traces} 
\label{sec:indep:traces}

In this section, we will develop an analog of Proposition~\ref{prop:disc:inde:pair} 
for $\varepsilon$-independent traces (action sequences) acting on neighboring states.

\subsection{$\varepsilon$-equivalent traces}
%\sayan{Why not use the term traces instead of action sequences?}
First, we define what it means for two finite traces in $A^*$ to be $\varepsilon$-equivalent.
\begin{definition}
	\label{por:def:equiv}
	For any $\varepsilon\geq 0$, we define a relation $R\subseteq A^*\times A^*$ such that 
	$\tau R \tau'$ iff there exists $\sigma,\eta\in A^*$ and $a,b\in A$ such that
	%\[
	$a\mysim\varepsilon b, \ \tau = \sigma a b \eta, \mbox{ and } \tau' = \sigma ba \eta.$
	%\]
	We define an equivalence relation $\myequiv\varepsilon \  \subseteq A^* \times A^*$ called {\em $\varepsilon$-equivalence}, as the reflexive and transitive closure of $R$.
\end{definition}
That is, two traces $\tau,\tau'\in A^*$ are $\varepsilon$-equivalent if we can construct $\tau'$ from $\tau$ by performing a sequence of swaps of consecutive $\varepsilon$-independent actions.

%It is not hard to see that $\myequiv \varepsilon$ is an equivalence relation on $A^*$. 
In the following proposition,
states that the last states of two potential executions starting from  the same initial discrete state (location) and resulting from equivalent traces have identical locations.
%\chuchu{The proof follows directly from Definition~\ref{por:def:equiv} and Proposition~\ref{prop:disc:inde:pair}(i) and is given in the full version~\cite{HM:TACAS-full}}.
\begin{proposition}
	\label{por:prop:perm_loc}
	Fix potential executions $\xi = \exe{q_0,\tau}$ and $\xi' = \exe{q_0',\tau'}$.
	%If $\xi$ and $\xi'$ have the same initial location and equivalent traces, then they have the same final location. That is,
	If $q_0.L = q_0'.L$  and $\tau\myequiv \varepsilon \tau'$, then  $\xi.\lstate.L = \xi'.\lstate.L$.
\end{proposition}
\begin{proof}
If $\tau = \tau'$, then the proposition follows from the Assumption.
Suppose $\tau \neq \tau'$, from Definition~\ref{por:def:equiv}, there exists a sequence of action sequences $\tau_0,\tau_1,\dots,\tau_k$ to join $\tau$ and $\tau'$ by swapping neighboring approximately independent actions.
Precisely the sequence $\{\tau_i\}_{i=0}^k$ satisfies:
\begin{inparaenum}[(i)]
\item $\tau_0 = \tau$ and $\tau_k = \tau'$, and 
\item for each pair $\tau_i$ and $\tau_{i+1}$, there exists $\sigma,\eta\in A^*$ and $a,b\in A$ such that
$a\mysim\varepsilon b$, $\tau_i = \sigma a b \eta$, and $\tau_{i+1} = \sigma ba \eta$.
\end{inparaenum}
From Definition~\ref{por:def:appind}, swapping approximately independent actions preserves the value of the discrete part of the final state. 
Hence for any $i\in[k]$, $\exe{q_0, \tau_i}.\lstate.L = \exe{q_0,\tau_{i+1}}.\lstate.L$. 
Therefore, $\xi.\lstate.L = \xi'.\lstate.L$.
\qed
\end{proof}

Next, we relate pairs of potential executions that result from $\varepsilon$-equivalent traces and  initial states that are $\delta$-close.
\begin{definition}
	\label{por:def:close}
	Given $\delta,\varepsilon \geq 0$, a pair of initial states $q_0, q_0'$, and 
	a pair traces $\tau, \tau' \in A^*$, 
	the corresponding potential executions $\xi = \exe{q_0,\tau}$ and $\xi'=\exe{q_0',\tau'}$ are {\em ($\delta,\varepsilon$)-related\/}, denoted by $\xi \myequivexe{\delta,\varepsilon} \xi'$, if 
	$q_0.L = q_0'.L$, $|q_0.X-q_0'.X| \leq \delta$, and $\tau \myequiv{\varepsilon} \tau'$.
\end{definition}
%That is, two potential executions $\xi,\xi'$ are ($\delta,\varepsilon$)-close if their first states lie in the $\delta$-neighborhood of each other, and their traces are $\varepsilon$-equivalent.
	\begin{example}
		\label{por:ex:linear3}
		In Example~\ref{por:ex:linear2}, we show that $a_0\mysim\varepsilon a_1$ and $a_0\mysim\varepsilon a_2$ with $\varepsilon = 0.1$.
		Consider the executions
		$
		\xi = q_0, a_0, q_1,a_1,q_2,a_2,q_3,a_\bot,q_4$
		and
		$\xi' = q_0', a_1,q_1',a_2,q_2',a_0,q_3',a_\bot,q_4'.
		$
		with traces $\trace\xi = a_0a_1a_2a_\bot$ and $\trace{\xi'} = a_1a_2a_0a_\bot$.
		For  $\varepsilon = 0.1$, we have $a_0a_1a_2a_\bot \myequiv \varepsilon  a_1a_0a_2a_\bot$ and $a_1a_0a_2a_\bot \myequiv\varepsilon  a_1a_2a_0a_\bot$.
		Since the equivalence relation $\myequiv\varepsilon$ is transitive, we have $\trace\xi \myequiv\varepsilon \trace{\xi'}$.
		Suppose $q_0\in\ball{\delta}{q_0'}$, then $\xi$ and $\xi'$ are ($\delta,\varepsilon$)-related executions with $\varepsilon=0.1$.
	\end{example}

It follows from Proposition~\ref{por:prop:perm_loc} that the discrete state (locations) reached by any pair of ($\delta,\varepsilon$)-related potential executions are the same.
%\chuchu{We first define }
 At the end of this section, in Lemma~\ref{por:lem:perm}, we will bound the distance between the continuous state reached by  ($\delta,\varepsilon$)-related potential executions.
\chuchu{We define in the following this bound as what we call {\em trace equivalent discrepancy factor ($\bloatf$)}, which is a constant number that works for all possible values of the variables starting from the initial set}.
Looking ahead, by bloating a single potential execution by the corresponding $\bloatf$, we can over-approximate the reachset of all related potential executions. This will be the basis for the reachability analysis in Section~\ref{por:sec:general}.

\begin{definition}
	\label{por:def:represent}
	For any potential execution $\xi$ and constants $\delta,\varepsilon \geq 0$,
	a {\em $(\delta, \varepsilon)$-trace equivalent \chuchu{discrepancy factor} ($\bloatf$)} is a nonnegative constant $r \geq 0$,
	such that for any $(\delta, \varepsilon)$-related potential finite execution $\xi'$,  
	\[
	|\xi'.\lstate.X -  \xi.\lstate.X| \leq r.
	\]
\end{definition}
%%%
That is, if $r$ is a $(\delta, \varepsilon)$-$\bloatf$, then the $r$-neighborhood of $\xi$'s last state $\B_r(\xi.\lstate)$ contains the last states of all other ($\delta,\varepsilon$)-related potential executions.
%\chuchu{Note that such upper bound $r$ and will be }
%\sayan{This example may have to be moved/ updated to show a basic $\bloatf$ calculation.}

\subsection{$(0,\varepsilon)$-trace equivalent discrepancy for traces (on the same initial states)}

In this section, we will develop an inductive method for computing $(\delta,\varepsilon)$-$\bloatf$. 
We begin by bounding the distance between potential executions that differ only in the position of a single action.
%We develop an upper bound of the distance for a simplified case, where the action being inserted is independent to all other actions.
%Throughout this section, we consider actions that are mutually $\varepsilon$-independent.
%We will quantify the distance between two potential executions with equivalent traces, 
%using the above $\gamma_n$  function.
%Consider an execution fragment $\exe{q_0,\tau}$ of the form,
%\[
%\xi = q_0, b_0, q_1, b_1, q_2, b_2, ..., b_{n-1}, q_n. 
%\]
%then for an action sequence $\tau'$ as an permutation of $\tau$, $\exe{q_0, \tau'}$ should be  a valid execution fragment and stays close to $\exe{q_0,\tau}$.
%Our first step involves computing the distance between two potential executions after inserting a single action at different positions. 
%First, we study a simple scenario  where the inserted action is independent from all others in the execution.
\begin{lemma}
\label{por:lem:insert}
Consider any $\varepsilon \geq 0$, an initial state $q_0\in Q$, an action $a\in A$ and a trace  $\tau \in A^*$ with $len(\tau) \geq 1$. 
If $\tau \mysim\varepsilon a$, then
the potential executions $\xi = \exe{q_0, \tau a}$ and $\xi' = \exe{q_0, a\tau}$ satisfy
%If $\xi = \exe{q_0, \tau}$ and $\xi' = \exe{q_0, a\tau}$ are both valid execution fragments, 
%and $a\in enabled(q_0)$, then 
\begin{enumerate}[(i)]
%\item $\xi' = \exe{q_0, a\tau}$ is also a valid execution fragment,
\item $\xi'.\lstate.L = \xi.\lstate.L$ and
\item $|\xi'.\lstate.X - \xi.\lstate.X| \leq \gamma_{n-1}(\varepsilon)$,
where $\gamma_n$ corresponds to the set of discrepancy functions $\{\beta_c\}_{c\in\tau}$ for the actions in $\tau$.
%, where $\gamma_n = \sum_{i \in [n]} \beta^i$ with $\beta = \max\{\beta_{\tau(0)}, \dots, \beta_{\tau(n-1)}, \beta_a\}$ is the upper-bound of the discrepancy functions.
\end{enumerate}
\end{lemma}
%\sayan{Check that definition $\gamma_n$ lines up properly. Check proof.}
\begin{proof}
%Parts $(i)$ and $(ii)$ are classic results for finite state systems, see for example Lemma 6.8 of~\cite{}.
%@book{baier2008principles,
%  title={Principles of model checking},
%  author={Baier, Christel and Katoen, Joost-Pieter and Larsen, Kim Guldstrand},
%  year={2008}
%}
Part (i) directly follows from Proposition~\ref{por:prop:perm_loc}.
We will prove part (ii) by induction on the length of  $\tau$. \\
{\bf Base}: For any trace $\tau$ of length 1, $\xi$ and $\xi'$ are of the form $\xi= q_0, b_0, q_1, a, q_2$ and $\xi'= q_0, a, q_1', b_0, q_2'$.
Since $a\mysim\varepsilon b_0$ and the two executions start from the same state, it follows from Definition~\ref{por:def:appind} that $|q_2'.X - q_2.X| \leq \varepsilon$. 
%\sayan{Definition of independent actions applies to identical states, but here $q_2$ and $q_2'$ may be different.?? Should this be about $q_1$?}
%
Recall from the preliminary that $\gamma_0(\varepsilon) = \beta^0(\varepsilon) = \varepsilon$. Hence $|q_2'.X - q_2.X| \leq \gamma_0(\varepsilon)$ holds for trace $\tau$ with $len(\tau) = 1$.\\
%%%%%%%%%%%
{\bf Induction}: Suppose the lemma holds for any $\tau$ with length at most $n-1$. 
Fixed any $\tau= b_0b_1\dots b_{n-1}$ of length $n$, we will show the lemma holds for $\tau$.

Let the potential executions $\xi=\exe{q_0, \tau a}$ and $\xi'=\exe{q_0, a\tau}$ be the form
\[
\begin{array}{c}
\xi=q_0,b_0,q_1, b_1, ..., b_{n-1}, q_n, a, q_{n + 1}, \\
\xi'= q_0, a, q_1', b_0, q_2', b_1,  ..., b_{n-1}, q_{n+1}'.
\end{array}
\]

\begin{figure}[H]
\centering
\includegraphics[width=.6\textwidth]{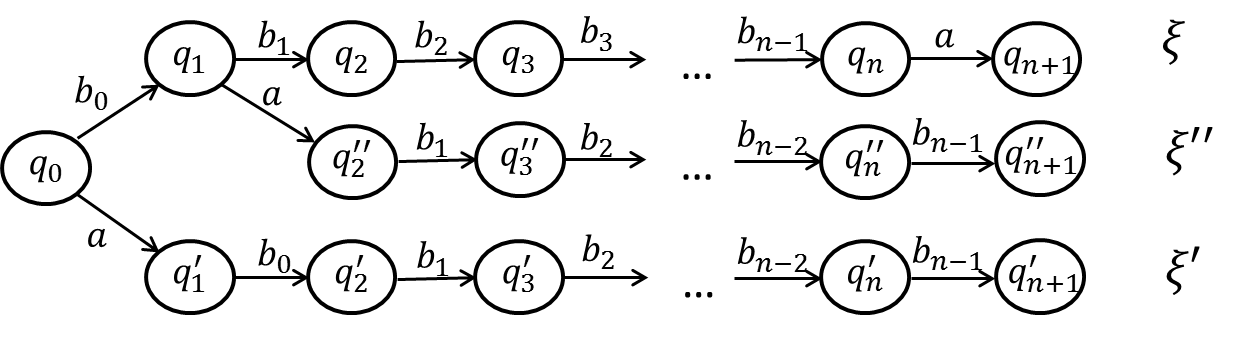}
\caption{\scriptsize Potential executions $\xi, \xi'$, and $\xi''$.}
\label{fig:insert}
\end{figure}

It suffices to prove that  $|\xi.\lstate.X-\xi'.\lstate.X| = |q_{n + 1}.X-q_{n + 1}'.X| \leq \gamma_{n-1}(\varepsilon)$.
We first construct a potential execution $\xi'' = \exe{q_0,b_0ab_1\dots b_{n-1}}$ by swapping the first two actions of $\xi'$. Then, $\xi''$ is of the form:
%\[
$\xi''= q_0, b_0, q_1, a, q_2'', b_1,  ..., b_{n-1}, q_{n+1}''.$
%\]
The potential executions $\xi,\xi'$ and $\xi''$ are shown in Figure~\ref{fig:insert}.
We first compare the potential executions $\xi$ and $\xi''$.
Notice that, $\xi$ and $\xi''$ share a common prefix $q_0, b_0, q_1$. 
Starting from $q_1$, the action sequence of $\xi''$ is derived from $\trace{\xi}$ 
by inserting action $a$ in front of the action sequence $\tau' =b_1b_2\dots b_{n-1}$. 

Since $\tau'\mysim{\varepsilon} a$, applying the induction hypothesis on the length $n-1$ action sequence $\tau'$, we get
%\begin{equation}
%\label{por:eq:lem11}
$|q_{n+1}.X - q_{n+1}''.X|\leq \gamma_{n-2}(\varepsilon).$
%\end{equation}
%
Then, we compare the potential executions $\xi'$ and $\xi''$.
Since $b_0\mysim\varepsilon a$, by applying the property of Definition~\ref{por:def:appind} to 
the first two actions of $\xi'$ and $\xi''$, we have
$|q_2'.X-q_2''.X| \leq \varepsilon$.
We note that $\xi'$ and $\xi''$ have the same suffix of action sequence from $q_2'$ and $q_2''$.
%It follows from Assumption~\ref{por:ass:main} that
%\begin{equation}
%\label{eq:lem12}
%q_{n+1}'.L = q_{n+1}''.L.
%\end{equation}
%Combining~\eqref{eq:lem11} and~\eqref{eq:lem12}, we have 
%\begin{equation}
%\label{eq:lem1main1}
%q_{n+1}.L = q_{n+1}''.L = q_{n+1}'.L.
%\end{equation}
Using Proposition~\ref{por:prop:discrep} from states $q_2'$ and $q_2''$, we have
\begin{equation}
\label{por:eq:lem13}
|q_{n+1}'.X-q_{n+1}''.X| \leq \beta_{b_1}\beta_{b_2}\dots\beta_{b_{n-1}}(|q_2'.X-q_2''.X|) \leq \beta^{n-1}(\varepsilon).
\end{equation}
Combining %~\eqref{por:eq:lem11}
the bound on $|q_2'.X-q_2''.X|$
 and~\eqref{por:eq:lem13} with triangular inequality, we have
$
|q_{n+1}.X - q_{n+1}'.X| \leq |q_{n+1}.X - q_{n+1}''.X| + |q_{n+1}'.X - q_{n+1}''.X| \leq \gamma_{n-2}(\varepsilon) + \beta^{n-1}(\varepsilon)
=\gamma_{n-1}(\varepsilon).
$
%Therefore, 
%the lemma holds for the action sequence $\tau$ with length $n$, which completes the induction.
\qed
\end{proof}

\subsection{$(\delta,\varepsilon)$-trace equivalent discrepancy for traces}

Lemma~\ref{por:lem:insert} gives a way to compute $(0,\varepsilon)$-$\bloatf$. Now, we generalize this  to compute $(\delta,\varepsilon)$-$\bloatf$, for ($\delta,\varepsilon$)-related  potential executions, with any $\delta \geq 0$. 
%(as defined in Definition~\ref{def:close}).
%
%The parameter $r$ is the {\em representative radius} of $\xi$.
\chuchu{The following lemma gives an inductive way of constructing $\bloatf$ as an action $a$ is appended to a trace $\tau$}.
%First, in Lemma~\ref{por:lem:perm} we start with a special case where the appended action $a$ is $\varepsilon$-independent of all other actions on $\xi$.

\begin{lemma}
\label{por:lem:perm}
For any  potential execution $\xi = \xi_{q_0,\tau}$
and constants $\delta,\varepsilon \geq 0$, 
if $r$ is a $(\delta, \varepsilon)$-$\bloatf$ for $\xi$, and the action $a\in A$ satisfies $\tau\mysim\varepsilon a$, then 
$
r' =\beta_a(r) + \gamma_{len(\tau)-1}(\varepsilon)
$
is a $(\delta, \varepsilon)$-$\bloatf$ for $\xi_{q_0,\tau a}$.
\end{lemma}

\begin{proof}
Fix any $\xi'$ that is $(\delta,\varepsilon)$-related to $\xi$ and with initial state $q_0'\in\ball{\delta}{q_0}$. It follows from Proposition~\ref{por:prop:perm_loc} that $\xi'.\lstate.L = \xi.\lstate.L$.
It suffices to prove that $|\xi'.\lstate.X - \xi.\lstate.X|\leq r'$.

Since $\trace{\xi'} \myequiv \varepsilon \tau a$, $\trace{\xi'}$ is in a form $\phi a \eta$ with some $\phi\eta\myequiv\varepsilon \tau$.
We construct a potential execution $\xi'' = \exe{q_0', \phi\eta a}$. The three potential executions are illustrated in Figure~\ref{fig:perm} below.
\begin{figure}[h!]
\centering
\includegraphics[width=0.4\textwidth]{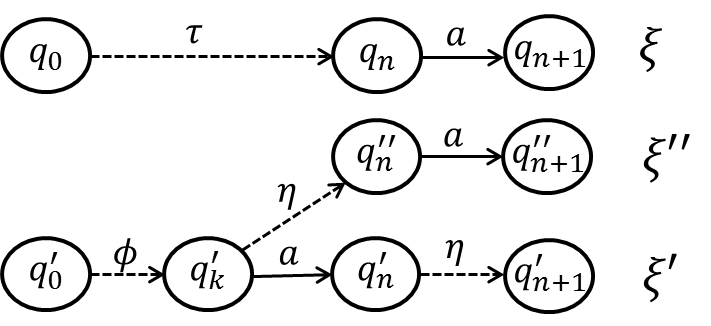}
\caption{Execution $\xi$, its $\varepsilon$-equivalent execution $\xi'$, and execution $\xi''$ that is constructed by swapping action $a$ to the back in $\xi'$.}
\label{fig:perm}
\end{figure}

We note that $r$ is a $\bloatf$ for the the prefix ($q_0,\tau,q_n$) of $\xi$ and $\delta,\varepsilon$.
Since $\phi\eta \myequiv \varepsilon \tau$ and $q_0'\in\ball{\delta}{q_0}$, it follows from Definition~\ref{por:def:represent} that $|q_n.X-q_n''.X|\leq r$. Hence
\begin{equation}
\label{por:eq:lem21}
|\xi.\lstate.X - \xi''.\lstate.X|\leq \beta_a(|q_n.X-q_n''.X|) \leq \beta_a(r).
\end{equation}
On the other hand, we observe that the traces of $\xi'$ and $\xi''$ differ only in the position of action $a$.
Application of Lemma~\ref{por:lem:insert} on $\xi'$ and $\xi''$ yields  
\begin{equation}
\label{por:eq:lem22}
|\xi'.\lstate.X - \xi''.\lstate.X|\leq \gamma_{len(\eta)-1}(\varepsilon)\leq \gamma_{len(\tau)-1}(\varepsilon).
\end{equation}
Combining~\eqref{por:eq:lem21} and~\eqref{por:eq:lem22} with triangular inequality, we have
\[
|\xi.\lstate.X - \xi'.\lstate.X| \leq \beta_a(r) + \gamma_{len(\tau)-1}(\varepsilon).
\]
%Therefore the lemma holds.
\qed
\end{proof}

\section{Reachability with approximate partial order reduction}
\label{por:sec:general}

We will present our main algorithm (Algorithm~\ref{alg:main}) for reachability analysis with approximate partial order reduction in this section. The core idea is to over-approximate $\reach{}{B_\delta(q_0),T}$ by (a) computing the actual execution $\xi_{q_0,\tau}$ and (b) expanding  this $\xi_{q_0,\tau}$ by a $(\delta, \varepsilon)$-$\bloatf$ to cover  all the states reachable from any other $(\delta_0,\varepsilon)$-related potential execution. 
%\sayan{This statement does not distinguish potential executions.}
% soundness
Combining such over-approximations from a cover of $\Theta$, we get over-approximations of $\reach{}{\Theta,T}$, and therefore, Algorithm~\ref{alg:main} can be used to soundly check for bounded safety or invariance. 
% completeness
The over-approximations can be made arbitrarily precise by  shrinking $\delta_0$ and $\varepsilon$ (Theorem~\ref{them:precision}).
Of course, at $\varepsilon = 0$ only traces that are exactly equivalent to $\tau$ will be covered, and nothing else. 
%
% efficiency
Algorithm~\ref{alg:main} avoids computing  $(\delta_0,\varepsilon)$-related executions, and therefore,  gains (possibly exponential) speedup.
 
The key subroutine in Algorithm~\ref{alg:main} is $\Post$ which computes the $\bloatf$ by adding one more action to the traces.
% for  $\xi_{q_0,\tau a}$, given the $\bloatf$ of $\xi_{q_0,\tau}$ and $\varepsilon$.
It turns out that, the $\bloatf$ is independent of $q_0$, but only depends on the sequence of actions in $\tau$. 
$\Post$ is used to compute  $\delta_{t}$ from $\delta_{t-1}$, such that, $\delta_t$ is the  $\bloatf$ for the length $t$ prefix of $\xi$. 
Let action $a$ be the $t^{th}$ action and $\xi = \xi_{q_0, \tau a}$. If  $a$ is $\varepsilon$-independent to $\tau$, then the $\bloatf$ $\delta_t$ can be computed from $\delta_{t-1}$ just using Lemma~\ref{por:lem:perm}.
%
%\sayan{Fixed lemma 2, please check.}
%
For the case where  $a$ is not  $\varepsilon$-independent to the whole sequence $\tau$, 
we would still want to compute a set of executions that $\xi_{q_0, \tau a}$ can cover. 
We observe that, with appropriate computation of $\bloatf$,
 $\xi_{q_0, \tau a}$ can cover all executions of the form $\xi_{q_0, \phi a \eta}$, 
 where  $\phi a \eta$ is $\varepsilon$-equivalent to $\tau a$ and $a \notin \eta$.
%
%the first (earliest) position that $a$ can occur in, among all the traces equivalent to $\tau$.
In what follows, we introduce this notion of {\em earliest equivalent position of $a$ in $\tau$\/} (Definition~\ref{def:anchor}), which is the basis for
the $\Post$ subroutine, which in turn is then used in the main reachability Algorithm~\ref{alg:main}.

\subsection{Earliest equivalent position of an action in a trace}
\label{por:sec:anchor}

%In the proof of Lemma~\ref{por:lem:perm}, we use the fact that any $\varepsilon$-equivalent trace of $\tau a$ must be in a form of $\phi a \eta$, where $\phi$ and $\eta$ are respectively the prefix and suffix of the last occurrence of action $a$.
%We observe that, the representative radius of $\exe{q_0, \tau a}$ depends on the length of $\phi$ and $\eta$. 
%For computing the representative radius, we need the maximum length of $\eta$, or equivalently the minimum length of $\phi$. We introduce the notion of {\em EEP} to capture this quantity.

For any trace $\tau\in A^*$ and  action $a \in \tau$,
we define $\mathit{lastPos}(\tau,a)$ as the largest index $k$ such that $\tau(k) = a$.
The earliest equivalent position, $\eep(\tau,a,\varepsilon)$ is the minimum of $\mathit{lastPos}(\tau',a)$ in any $\tau'$ that is $\varepsilon$-equivalent to $\tau a$.
\begin{definition}
\label{def:anchor}
For any trace $\tau\in A^*$, $a\in A$, and $\varepsilon >0$, 
the {\em earliest equivalent position} of $a$ on $\tau$ is 
$
\eep(\tau,a,\varepsilon) \deq \min_{\tau' \myequiv \varepsilon \tau a} \mathit{lastPos}(\tau',a).
$
\end{definition}
%\begin{definition}
%\label{def:anchor}
%For any action sequence $\tau\in A^*$ and any action $a\in A$,
%the {\em EEP} of $a$ on $\tau$ is the largest $k\in [len(\tau)]$ such that
%(i) $\tau(k)\not\mysim \varepsilon a$, and
%(ii) $\tau(k) \not \mysim \varepsilon \tau(l)$ for all $l < k$.
%If such $k$ does not exists, the EEP is defined as $-1$.
%For $k \geq 0$, we call the action $b = \tau(k)$ the {\em anchor action}.
%\end{definition}
%
%That is, the EEP of $a$ on $\tau$ is the position of the last action $b$ on $\tau$ such that
%$b$ is not approximately independent to both $a$ and any actions on $\tau$ before $b$.
%Such action $b$ can be seen as a barrier of $a$ on any 
For any trace $\tau a$, its $\varepsilon$-equivalent traces can be derived by swapping consecutive $\varepsilon$-independent action pairs.
Hence, the eep of $a$ is the leftmost position it can be swapped to, starting from the end.
Any equivalent trace of $\tau a$ is of the form $\phi a \eta$ where $\phi$ and $\eta$ are the prefix and suffix of the last occurrence of action $a$.
Hence,  equivalently:
$
\eep(\tau,a,\epsilon) = \min_{\phi a \eta \myequiv \varepsilon \tau a,\ a\notin \eta} len(\phi).
$
%
%In the following algorithm, we find the EEP of action $a$ on an action sequence $\tau$.
%
\chuchu{In Appendix~\ref{app:eep}} we give a simple $O(len(\tau)^2)$ algorithm  for computing $\eep()$. If the $\varepsilon$-independence relation is symmetric, then it $\eep$ can be computed in $O(len(\tau))$ time.

%\chuchu{
\begin{example}
In Example~\ref{por:ex:linear2}, we showed that $a_0\mysim\varepsilon a_1$ and $a_0\mysim\varepsilon a_2$ with $\varepsilon = 0.1$;  $a_\bot$ is not $\varepsilon$-independent to any actions.
What is $\eep(  a_\bot a_0 a_1,a_2, \varepsilon)$?
We can swap $a_2$ ahead following the sequence $\tau a_2 = a_\bot a_0 a_1 a_2\myequiv \varepsilon a_\bot a_1 a_0 a_2 \myequiv \varepsilon a_\bot a_1 a_2 a_0$.
As $a_\bot$ and $a_1$  are not independent of $a_2$, it cannot occur earlier.
$\eep(  a_\bot a_0 a_1,a_2, \varepsilon) = 2$.
%, corresponding to its position on $a_\bot a_1 a_2 a_0$.
%Algorithm~\ref{por:algo:anchor} construct a trace $\phi = a_\bot a_1$, where the length of $\phi$ gives the correct EEP $2$.
\end{example}
%}

\subsection{Reachability  using $(\delta,\varepsilon)$-trace equivalent discrepancy}

%In this section, we first introduce an algorithm (Algorithm \ref{alg:post} ) to compute the $\bloatf$ by adding one more action $a$ to a trace $\tau$, which is similar to Lemma \ref{por:lem:perm} but does not require that $\tau \mysim\varepsilon a$. Then we will introduce a complete algorithm (Algorithm \ref{alg:main}) that uses Algorithm \ref{alg:post} to compute an over-approximation of $\reach{}{\Theta, T}$.

$\Post$ (Algorithm~\ref{alg:post}) takes inputs of trace $\tau$, a new action to be added $a$, a parameter $r \geq 0$ such that $r$ is a $(\delta_0,\varepsilon)$-$\bloatf$ for the potential execution $\xi_{q_0,\tau}$ for some initial state $q_0$,  initial set radius $\delta_0$, approximation parameter $\varepsilon \geq 0$, and a set of discrepancy functions $\{\beta_a\}_{a\in A}$. It returns a $(\delta_0, \varepsilon)$-$\bloatf$ $r'$ for the potential execution $\xi_{q_0, \tau a}$.

\begin{algorithm}[h!]
%\scriptsize
\caption{$\Post(\tau, a, r, \varepsilon, \{\beta_a\}_{a\in A})$}
\label{alg:post}
\begin{algorithmic}[1]
%\Procedure{Bloat}{$\xi, \delta_0, \varepsilon, \beta$}
\State $\beta \gets \max_{b\in \tau a}\{\beta_b\}$; $k\gets\mathit{\eep(\tau, a, \epsilon)}$; $t \gets len(\tau)$;
\IfThenElse {$k = t$} %if
{$r' \gets \beta_{a}(r)$}% .. then 
{$r' \gets \beta_a(r) +  \gamma_{t-k-1}(\varepsilon)$} \label{ln:case1} %... else
%\If{$k = t$} 
%	\State $r' \gets \beta_{a}(r)$;
%	\label{ln:case1}
%		%\Comment{$a$ cannot be swapped}
%\Else  
		%\State $\beta \gets \max_{a\in \tau}\{\beta_a\}$
%	\State $r' \gets \beta_a(r) +  \gamma_{t-k-1}(\varepsilon)$;	
%	\label{ln:case2}
		%\Comment{$a$ can be swapped to a position $>k$}
%\EndIf
\State \Return $r'$;
\end{algorithmic}
\end{algorithm}

%We establish the correctness of Algorithm~\ref{alg:post} by showing that $r'$ is indeed an $\bloatf$ for the potential execution $\xi_{q_0, \tau a}$ and $\delta_0, \varepsilon$.
\begin{lemma}
\label{por:lem:sound}
For some initial state $q_0$ and initial set size $\delta_0$,
if $r$ is a $(\delta_0, \varepsilon)$-$\bloatf$ for $\xi_{q_0, \tau}$  then
value returned by $\Post()$ is a $(\delta_0, \varepsilon)$-$\bloatf$ for $\xi_{q_0,\tau a}$.
\end{lemma}

\begin{proof}
Let us fix some initial state $q_0$ and initial set size $\delta_0$.

Let  $\xi_{t} = \xi_{q_0, \tau}$ be the potential execution starting from $q_0$ by taking the trace $\tau$, and $\xi_{t+1} = \exe{q_0, \tau a}$.
Fix any $\xi'$ that is ($\delta_0,\varepsilon$)-related to $\xi_{t+1}$.
From Proposition~\ref{por:prop:perm_loc}, $\xi'.\lstate.L = \xi_{t+1}.\lstate.L$.
It suffice to prove that $|\xi'.\lstate.X - \xi_{t+1}.\lstate.X| \leq r'$.

\begin{figure}[tbhp!]
	\centering
	\includegraphics[trim={0 0 0 0},clip, width=0.45\textwidth]{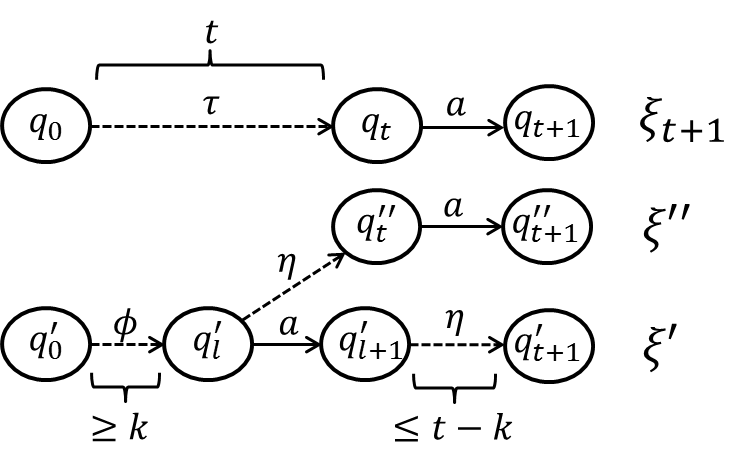}
	\caption{Potential executions $\xi_{t+1}$, $\xi'$,$\xi''$ }% constructed by swapping action $a$ in $\xi'$.}
	\label{fig:exes_algo}
\end{figure}

Since $\trace{\xi'} \myequiv \varepsilon \tau a$, action $a$ is in the sequence $\trace{\xi'}$.
Partitioning $\trace{\xi'}$ on the last occurrence of $a$, we get $\trace{\xi'} =\phi a \eta$ for some $\phi,\eta\in A^*$ with $a\not\in \eta$.
Since $k$ is the $eep$, from Definition~\ref{def:anchor}, the position of the  last occurrence of $a$ on $\trace{\xi'}$ is at least $k$.
Hence we have $len(\phi) \geq k$ and $len(\eta) = t - len(\phi) \leq t-k$.
We construct another potential execution $\xi'' = \exe{q'_0, \phi\eta a}$ with the same initial state as $\xi'$.
The executions $\xi_{t+1},\xi'$ and $\xi''$ are illustrated in Figure~\ref{fig:exes_algo}.
$q_t$ is the last state of the execution $\xi_{t}$. From the assumption, $\ball{r}{q_t}$ is an over-approximation of
the reachset at step $t$.
We note that the length $t$ prefix $\xi''$ is ($\delta_0,\varepsilon$)-related to $\xi_t$. 
Therefore, $|q_t.X - q_t''.X| \leq r$. Using the discrepancy function of action $a$, we have
\begin{equation}
\label{por:eq:algo1}
|q_{t+1}.X - q_{t+1}''.X| \leq \beta_a(|q_t.X - q_t''.X|) \leq \beta_a(r).
\end{equation}
We will quantify the distance between $\xi'$ and $\xi''$.
There are two cases: \\
(i) If $k = t$ then, $len(\eta) \leq t-k = 0$, that is, $\eta$ is an empty string. 
Hence, $\xi'$ and $\xi''$ are indeed identical and $q'_{t+1} = q''_{t+1}$. 
Thus from~\eqref{por:eq:algo1},
$
|q_{t+1}.X-q'_{t+1}.X| = |q_{t+1}.X - q_{t+1}''.X| \leq \beta_a(r),
$ 
and the lemma holds.
(ii) \chuchu{Otherwise $k < t$ and from Lemma~\ref{por:lem:insert}, we can bound the distance between $\xi'$ and $\xi''$ as
%\[
$|q_{t+1}'.X-q_{t+1}''.X| \leq \gamma_{len(\eta)-1}(\varepsilon) \leq \gamma_{t-k-1}(\varepsilon).$
%\]
Combining with~\eqref{por:eq:algo1}, we get
$
|q_{t+1}.X-q'_{t+1}.X| \leq |q_{t+1}.X-q''_{t+1}.X| + |q'_{t+1}.X-q''_{t+1}.X| \leq 
\beta_a(r) + \gamma_{t-k-1}(\varepsilon).
$}
%Therefore if $r'$ is computed by line~\ref{ln:case2}, the lemma holds.
\qed
\end{proof}

%\subsection{reachset over-approximation algorithm}
Next, we present the main reachability algorithm which uses $\Post$. 
%Algorithm \ref{alg:post} to eliminate the consideration of all executions that are $(\delta_0, \varepsilon)$-related the representative executions explored.
Algorithm~\ref{alg:main} takes inputs of an initial set $\Theta$, time horizon $T$, two parameters $\delta_0, \varepsilon \geq 0$, and a set of discrepancy functions $\{\beta_a\}_{a\in A}$. It returns the 
over-approximation of the reach set for each time step.

The algorithm first computes a $\delta_0$-cover $Q_0$ of the initial set $\Theta$ such that 
%the union of the $\delta_0$ balls around all initial states $q_0 \in Q_0$ covers the initial set $\Theta$: 
$\Theta \subseteq \cup_{q_0 \in Q_0}\ball{\delta}{q_0}$ (Line \ref{alg:main:initial}). 
The {\bf for}-loop from Line \ref{alg:main:bigout} to Line \ref{alg:main:endbigout} will compute the over-approximation of the reachset from each initial cover $\reach{}{\ball{\delta_0}{q_0}, t}$.  The over-approximation from each cover is represented as a collection $  \langle \RT_{0},\dots, \RT_{T} \rangle $, where each $\RT_{t} $ is a set of tuples $\langle \tau_t, q_t, \delta_t \rangle $ such that 
\begin{inparaenum}[(i)]
\item the traces $\RT_t \restr 1$ and their $\varepsilon$-equivalent traces contain the traces of all valid executions of length $t$,
\item the traces in $\RT_t \restr 1$ are mutually non-$\varepsilon$-equivalent,
\item for each tuple $\delta_{t}$ is the $(\delta_0, \varepsilon)$-$\bloatf$ for $\xi_{q_0, \tau_t}$,
\end{inparaenum}
%It can be concluded from above
% the reachset starting from the initial cover $\ball{\delta_0}{q_0}$ after the $t^{th}$ action is contained in the union of each $\delta_{t}$ balls around $q_{t}$: 
%$ \reach{}{\ball{\delta_0}{q_0}, t}$ is contained in the union $\cup_{\langle \tau_{t}, q_{t}, \delta_{t}  \rangle \in \RT_{t}} \ball{\delta_{t}}{q_{t}}$, for any $t = 0,\dots,T$. 

 For each initial cover $\ball{\delta_0}{q_0}$, $R_0$ is initialized as the tuple of empty string, the initial state $q_0$ and size $\delta_0$ (Line \ref{alg:main:R0}). Then the reachset over-approximation is computed recursively for each time step 
%At the beginning of the $t^{th}$-loop at Line \ref{alg:main:eachtime}, the union of all the traces $\RT_t \restr 1$ and their $\varepsilon$-equivalent traces include the traces of all valid length $t$ executions, and all $\tau_t$s should be mutually non-$\varepsilon$-equivalent. For each tuple  $\langle \tau_{t}, q_{t}, \delta_{t}  \rangle \in \RT_{t} $, $\delta_{t}$ is the $\bloatf$ for $\xi_{q_0, \tau_t}$ and $\delta_0, \varepsilon$. 
by checking for the maximum set of enabled actions $\EA$ for the set of states $\ball{\delta_{t}}{q_{t}}$ (Line \ref{alg:main:ea}), and try to attach each enabled action $a \in \EA$ to $\tau_t$ unless $\tau_t a$ is $\varepsilon$-equivalent to some length $t+1$ trace that is already in $\RT_{t+1} \restr 1$. This is where the major reduction happens using approximate partial order reduction. If not, the $(\delta_0, \varepsilon)$-$\bloatf$ for $\xi_{q_0, \tau_t a}$ will be computed using $\Post$, and new tuple  $\langle \tau_t a, q_{t+1}, \delta_{t+1} \rangle$ will be added to $\RT_{t+1}$ (Line \ref{alg:main:newaction}). 

If there are $k$ actions in total and they are mutually $\varepsilon$-independent, then as long as the numbers of each action in $\tau_t$ and $\tau'_t$ are the same, $\tau_t \myequiv\varepsilon \tau'_t$. Therefore, in this case, $\RT_{t}$ contains at most $\binom{t+k-1} {k-1}$ tuples. Furthermore, for any length $t$ trace $\tau_t$, if all actions in $\tau_t$ are mutually $\varepsilon$-independent, the algorithm can reduce the number of executions explored by $O(t!)$.
Essentially, each $\tau_t \in \RT_t \restr 1$ is a representative trace for the length $t$ $\varepsilon$-equivalence class.

\begin{algorithm}[th!]
\caption{Reachability algorithm to over-approximate $\reach{}{\Theta,T}$}
\label{alg:main}
\begin{algorithmic}[1]
\State {\bf Input}: $\Theta, T, \varepsilon, \delta_0, \{\beta_a\}$;	
\State $Q_0 \gets \delta_0$-$cover(\Theta)$; $\R \gets \emptyset$ \label{alg:main:initial}
\For{$q_0\in Q_0$} \label{alg:main:bigout}
	\State $\RT_0 \gets \{\langle '', q_0, \delta_0 \rangle\}$; \label{alg:main:R0}
	\For{$t = [T]$} \label{alg:main:eachtime}
		\State $\RT_{T} \gets \emptyset$;
		\For{ each $\langle \tau_t, q_t, \delta_t \rangle \in \RT_{t}$}
			\State $\EA \gets enabledactions(\ball{\delta_t}{q_t})$; \label{alg:main:ea}
			\For{$a \in EA$}
				\If{$\forall \tau_{t+1} \in R_{t+1} \restr 1, \neg \left(\tau_t a \myequiv \varepsilon \tau_{t+1} \right)$}; \label{alg:main:if}
						\State $q_{t+1} \gets a(q_t)$ \label{alg:main:newstate}
						\State $\delta_{t+1} \gets \Post(\tau_t, a, \delta_t, \varepsilon, \{\beta_a\}_{a\in A})$ \label{alg:main:newdelta}
						\State $\RT_{t+1} \gets \RT_{t+1} \cup \langle \tau_t a, q_{t+1}, \delta_{t+1} \rangle $ \label{alg:main:newaction}
				\EndIf \label{alg:main:endif}
			\EndFor
		\EndFor
	\EndFor \label{alg:main:eachtimeend}
	\State $\R \gets \R \cup \langle \RT_{0},\dots, \RT_{T} \rangle$
\EndFor \label{alg:main:endbigout}
\State \Return $\R$;  
\end{algorithmic}
\end{algorithm}

Theorem \ref{them:soundness} shows that Algorithm \ref{them:soundness} indeed computes an over-approximation for the reachsets, and Theorem \ref{them:precision} states that the over-approximation can be made arbitrarily precise by reducing the size of $\delta_0, \varepsilon$.

\begin{theorem}[Soundness]
\label{them:soundness}
Set $\R$ returned by Algorithm~\ref{alg:main}, satisfies $\forall t = 0,\dots,T,$ 
\begin{align}
\label{eq:sound}
\reach{}{\Theta, t} \subseteq \bigcup_{\RT_{t} \in \R \restr t} \bigcup_{\langle \tau, q, \delta  \rangle \in \RT_{t}} \ball{\delta}{q}.
\end{align}
\end{theorem}

\begin{proof}
Since $\cup_{q_0 \in Q_0}\ball{\delta}{q_0} \supseteq \Theta$, it suffices to show that at each time step $t = 0,\dots,T$,  the $\RT_{t}$ computed in the {\bf for}-loop from Line \ref{alg:main:R0} to Line \ref{alg:main:eachtimeend} satisfy $\reach{}{\ball{\delta_0}{q_0}, t} \subseteq \cup_{\langle \tau, q, \delta  \rangle \in \RT_{t}} \ball{\delta}{q}$. Fix any $q_0 \in Q_0$, we will prove by induction.

Base case: initially before any action happens, the only valid trace is the empty string $''$ and the initial set is indeed $\ball{\delta_0}{q_0}$.

Induction step: assume that at time step $ t < T$, the union of all the traces $R_{t} \restr 1$ and their $\varepsilon$-equivalent traces contain the traces of all length $t$ valid executions, and for each tuple  $\langle \tau_{t}, q_{t}, \delta_{t}  \rangle \in \RT_{t}
$, $\delta_t$ is a $(\delta_0, \varepsilon)$-$\bloatf$  for $\xi_{q_0, \tau_t}$. That is, $\ball{\delta_{t}}{q_{t}}$ contains the final states of all ($\delta_0, \varepsilon$)-related executions to $\xi_{q_0, \tau_t}$. This is sufficient for showing that $\reach{}{\ball{\delta_0}{q_0}, t}
 \subseteq \cup_{\langle \tau, q, \delta  \rangle \in \RT_{t}} \ball{\delta}{q}$.
 
 Since for each tuple contained in $\RT_t$, we will consider the maximum possible set of actions enabled at Line \ref{alg:main:ea} and attempts to compute
 the $(\delta_0, \varepsilon)$-$\bloatf$ for $\xi_{q_0, \tau_t a}$. If $\tau_t a$ is not  $\varepsilon$-equivalent to any of the length $t+1$ traces that has already been added to $\RT_{t+1}$, then Lemma \ref{por:lem:sound}
 guarantees that the $q_{t+1}$ and $\delta_{t+1}$ computed at Line \ref{alg:main:newstate} and \ref{alg:main:newdelta} satisfy that $\delta_{t+1}$ is the $(\delta_0, \varepsilon)$-$\bloatf$ for $\xi_{q_0, \tau_t a}$. 
Otherwise, $\tau_t a$ is $\varepsilon$-equivalent to some trace $\tau_{t+1}$ that has already been added to $\RT_{t+1}$, then for any initial state $q_0'$ that is $\delta_0$-close to $q_0$, $\xi_{q_0', \tau_t a}$ and $\xi_{q_0,\tau_{t+1}}$ are $(\delta_0,\varepsilon)$-related and the final state of $\xi_{q'_0, \tau_t a}$ is already contained in $\ball{\delta_{t+1}}{q_{t+1}}$. Therefore,  the union of all the traces $R_{t+1} \restr 1$ and their $\varepsilon$-equivalent traces contain the traces of all length $t+1$ valid executions, and for each tuple  $\langle \tau_{t+1}, q_{t+1}, \delta_{t+1}  \rangle \in \RT_{t+1}$, $\delta_{t+1}$ is a $(\delta_0, \varepsilon)$-$\bloatf$ for $\xi_{q_0, \tau_{t+1}}$, which means  ${\reach{}{\ball{\delta_0}{q_0}, t+1}} \subseteq \cup_{\langle \tau, q, \delta  \rangle \in \RT_{t+1}} \ball{\delta}{q}$. So the theorem holds.
\qed
\end{proof}

\begin{theorem}[Precision]
\label{them:precision}
For any $r>0$, there exist $\delta_0,\varepsilon>0$ such that, the reachset over-approximation $\R$ computed by Algorithm~\ref{alg:main} satisfies $\forall t = 0,\dots,T,$ 
\begin{equation}
\bigcup_{\RT_{t} \in \R \restr t} \bigcup_{\langle \tau, q, \delta  \rangle \in \RT_{t}} \ball{\delta}{q} \subseteq\ball{r}{\reach{}{\Theta,t}}.
\end{equation}
\end{theorem}

\begin{proof}
From Proposition~\ref{por:prop:gamma}, 
for any $n$, $\gamma_n(\varepsilon)\rightarrow 0$ as $\varepsilon \rightarrow 0$.
From Definition~\ref{por:def:discrep}, for any $\delta_t$ and discrepancy function $\beta$, $\beta(\delta_t) \rightarrow 0$ as $\delta_t\rightarrow 0$.
Therefore, when Line \ref{alg:main:newdelta} of Algorithm~\ref{alg:main} is executed, $\delta_{t+1} \rightarrow 0$ as $\delta_t\rightarrow 0$ and $\varepsilon\rightarrow 0$.
Iteratively applying this observation leads that $\delta_t$ contained in any set $\RT_t$ converges to zero as $\delta_0\rightarrow 0$ and $\varepsilon \rightarrow 0$. 

%From the proof of Theorem \ref{them:soundness}, we know that $\reach{}{\ball{\delta_0}{q_0}, T}
% \subseteq \cup_{\langle \tau_{T}, q_{T}, \delta_{T}  \rangle \in \RT_{T}} \ball{\delta_{T}}{q_{T}}$.

Fix arbitrary $r>0$.
The set $\R$ is a union of approximations for each $\reach{}{\ball{\delta_0}{q_0}, T}$.
Fix any such $q_0, \delta_0$, it suffices to show that $\cup_{\langle \tau, q, \delta  \rangle \in \RT_{t}} \ball{\delta}{q} \subseteq \ball{r}{\reach{}{\Theta,t}}$ for small enough $\delta_0$ and $\varepsilon$. Moreover, it suffices to show that fix any $\langle \tau_t, q_t, \delta_t  \rangle \in \RT_{t}$,  $\ball{\delta}{q_t} \subseteq \ball{r}{\reach{}{\Theta,t}}$ for small enough $\delta_0$ and $\varepsilon$.

Since each $\delta_t$ is a $(\delta_0, \varepsilon)$-$\bloatf$ of the execution $\xi_{q_0,\tau_t}$ and $\delta_0,\varepsilon$, there is an execution $\xi'= \xi'_{q'_0,\tau'}$ from $q'_0 \in \ball{\delta}{q_0}$ following the trace $\tau' \myequiv\varepsilon \tau_t$.
By the definition of reachset, we have $\xi'(t) \in \reach{}{\Theta,t}$.
On the other hand, $\xi'$ is ($\delta_0,\varepsilon$)-related to the potential execution $\xi_{q_0,\tau_t}$, so $\xi'(t) \in \ball{\delta_t}{q_t}$.
That is, $\ball{\delta_t}{q_t}$ and the reachset $\reach{}{\Theta,t}$ has intersections at the state $\xi'(t)$.

The radius of at each time step $\delta_t$ can be made arbitrarily small as $\delta_0$ and $\varepsilon$ go to $0$.
We chose small enough $\delta_0$ and $\varepsilon$, such that the radius of $ \ball{\delta_{t}}{q_{t}}$ is less than $r/2$.
Therefore, $ \ball{\delta_{t}}{q_{t}}$  is contained in the radius $r$ ball of the reachset $\ball{r}{\reach{}{\Theta,t}}$.
\qed
\end{proof}
%The proof is based on the fact that $\delta_t$ computed in Algorithm \ref{alg:main} converges to zero as $\delta_0\rightarrow 0$ and $\varepsilon \rightarrow 0$, \chuchu{and the details are given in the full version of the paper~\cite{HM:TACAS-full}}.
Notice that as $\delta_0$ and $\varepsilon$ go to $0$, the Algorithm \ref{alg:main} actually converges to a simulation algorithm which simulates every valid execution from a single  initial state.

\section{Experimental evaluation of effectiveness}
\label{por:sec:case}

We  discuss the results from evaluating Algorithm~\ref{alg:main} in three case studies. Our Python implementation runs on a standard laptop (Intel Core\textsuperscript{TM} i7-7600 U CPU, 16G RAM).
%
%\sayan{How much time does it take to run all executions?}
%Are we reporing running time or memory usage? Mention CPU?}

%\sayan{Across this section let us consistently use $x[i]$ instead of $x_i$. There is lot of mixing now.}

\paragraph{Iterative consensus.}
This is an instance of  $\auto{Consensus}$ (Example~\ref{por:ex:linear}) with  $3$ continuous variables and 3 actions $a_0,a_1,a_2$. 
%and showed that $a_0\mysim\varepsilon a_2$ and $a_1 \mysim \varepsilon a_2$.
%In Example~\ref{por:ex:linear3} we further presented an execution and its ($\delta,\varepsilon$)-close executions.
We want to check if the continuous states  converge to $[-0.4,0.4]^3$ in 3 rounds
starting from a radius \chuchu{$0.5$} ball around $[2.5,0.5,-3]$.
% soundness
Figure~\ref{fig:linear_ex} (Left) shows reachset over-approximation 
computed  and projected on $x[0]$. The blue and red curves give the bounds.  
As the figure shows, $x[0]$ converges to  $[-0.4, 0.4]$ at round 3; and so do $x[1]$ and $x[2]$ (not shown).
%\sayan{2. What does the reachset computation say about the convergence property?}
We also  simulated $100$ random valid executions (yellow curves) from the initial set and validate that indeed the over-approximation is sound. 
% efficientcy

Recall, three actions can occur in any order in each round, i.e., $3! = 6$ traces per round, and $6^{3} = 216$  executions from a single initial state up to $3$ rounds.
We showed in Example~\ref{por:ex:linear2} that  $a_0\mysim\varepsilon a_1$ and $a_0\mysim\varepsilon a_2$ with $\varepsilon=0.1$. Therefore, $a_0a_1a_2 \myequiv\varepsilon a_1a_0a_2 \myequiv\varepsilon a_1a_2a_0$ and $a_0a_2a_1 \myequiv\varepsilon a_2a_0a_1 \myequiv\varepsilon a_2a_1a_0$, and Algorithm~\ref{alg:main}  explored only $2$ (length $12$) executions from a set of initial states for computing the bounds. The running time for Algorithm \ref{alg:main} is 1 millisecond while exploring  all valid executions from even only a single  state took 20 milliseconds.
%we compute a tube around the execution to over-approximate all ($\delta,\varepsilon$)-close potential executions (green curves).
%A $\delta,\varepsilon$-close potential execution  $\xi'=\exe{q_0',\tau}$ with
%$q_0'.x = [2.3, -3.2, 1]$ and $\tau' = a_0a_1a_2a_\bot a_1a_0a_2 a_\bot a_2a_1a_0 a_\bot$ is shown by the red curve.
%The result validates that $\xi'$ lies in the tube computed using our technique.

\begin{figure}[bpht!]
\centering
\includegraphics[width=.4\textwidth]{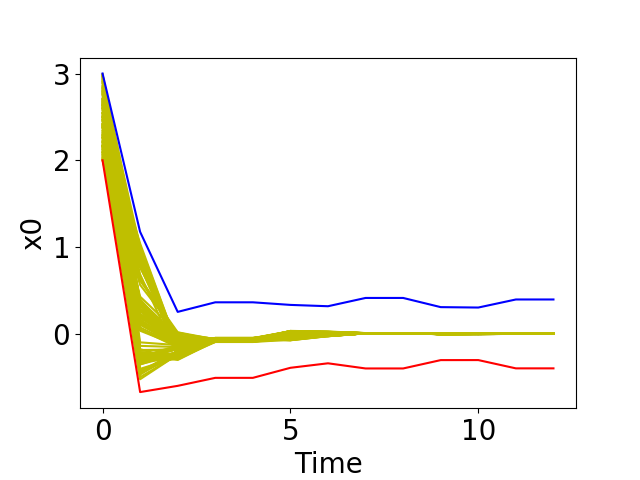}
\includegraphics[width=0.4\textwidth]{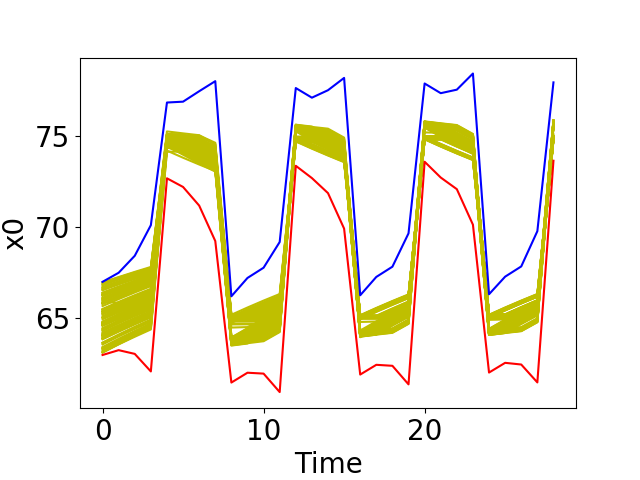}
\caption{\scriptsize Reachset computations. The blue curves are the upper bound of the reachsets and the red curves are the lower bound of the reachsets. Between the blue and red curves, the yellow curves are $100$ random simulations of valid executions.
{\em Left:} Linear transition system. {\em Right:} Room heating system.}
\label{fig:linear_ex}
\end{figure}
\paragraph{Platoon.}
Consider an $N$ car platoon on a single lane  (\chuchu{see Figure~7 in Appendix~\ref{app:platoon} for the pseudocode and details}). Each car can choose one of three actions at each time step: $a$ (accelerate), $b$ (brake), or $c$ (cruise). Car $0$ can choose any action at each time step; remaining cars try to keep safe distance with predecessor by choosing accelerate ($a$) if the distance is more than $50$,  brake ($b$) if the distance is less than $30$, and cruise ($c$) otherwise. 

Consider a 2-car platoon and a time horizon of $T =10$. We want to verify that the cars maintain  safe separation. Reachset over-approximations projected on the position variables are shown in Figure \ref{fig:cars}, with $100$ random simulations of valid executions as a sanity check.
Car 0 has lots of choices  and it's position over-approximation diverges (Figure~\ref{fig:cars}). Car 1's position depends on its initial relative distance with Car 0. It is also easy to conclude from Figure~\ref{fig:cars} that two cars maintain safe relative distance for these different initial states.

From a single initial state, in every step, Car 0 has $3$ choices, and therefore there are $3^{10}$ possible executions. Considering a range of initial positions for two cars, there are infinitely many execution, and $9^{10}$ (around 206 trillion) possible traces.
With $\epsilon = 0.282$, Algorithm~\ref{alg:main} explored a maximum of $\binom{18}{8} = 43758$ traces; the concrete number varies for  different initial sets. The running time for Algorithm \ref{alg:main} is 5.1 milliseconds while exploring all valid executions from even only a single  state took 2.9 seconds.

For a 4-car platoon and a time horizon of $T =10$,
%. We want to verify that all the cars maintain  safe separation. 
%%
%From a single initial state, in every step, Car 0 has $3$ choices, and therefore there are $3^{10}$ possible executions. 
%Considering a range of initial positions for 4 cars, 
there are  $81^{10}$ possible traces considering a range of initial positions.
With $\varepsilon = 0.282$, Algorithm~\ref{alg:main} explored $7986$ traces to conclude that all cars maintain safe separation for the setting where all cars are initially separated by a distance of $40$ and has an initial set radius of $4$. The running time for Algorithm \ref{alg:main} is 62.3 milliseconds, while exploring all valid executions from even only a single state took 6.2 seconds.

%\sayan{1. why $\binom{12}{8}$? 2. this is not directed towards verifying a property and so does not quite line up with the first example?}
%
%Both cars will start with velocity $20$, Car0 has the freedom to choose any action at each time step. Car1 tries to keep the distance with Car0 in some range by accelerates if its distance with Car0 is more than $50$,  brakes if the distance is less than $30$ and cruises otherwise. 
%We report the over-approximation of the reachset for the position of the cars in Figure \ref{fig:cars} with different initial positions for the cars for $10$ steps. 
%
%
%Without the partial order reduction, there could be $9^{10}$ (around 3.4 billion) valid traces. However, with $\epsilon = 0.282$, Algorithm \ref{alg:main} will explore at most $\binom{12}{8} = 495$ traces, and the concrete number will depend on the different initial sets.

\begin{figure}[bpht!]
\centering
\includegraphics[width=0.32\textwidth]{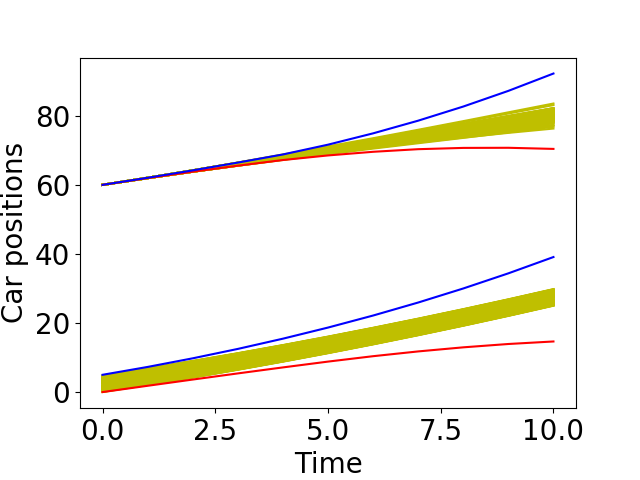}
\includegraphics[width=0.32\textwidth]{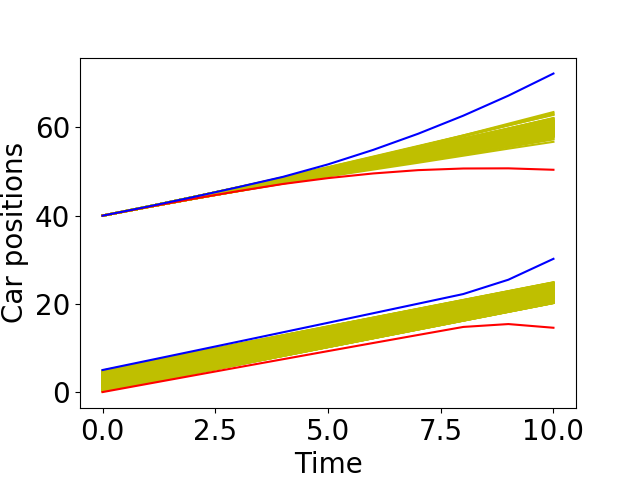}
\includegraphics[width=0.32\textwidth]{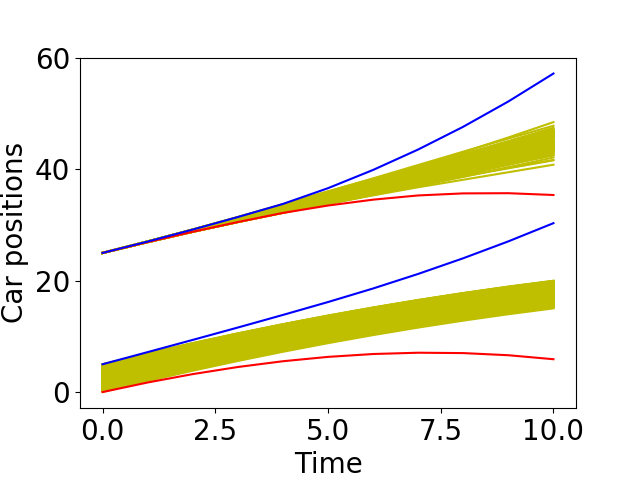}
\caption{\scriptsize Position over-approximations for 2 cars. The blue curves are the upper bound of the reachsets and the red curves are the lower bound of the reachsets. Between the blue and red curves, the yellow curves are $100$ random simulations of valid executions. Car1's initial position is in the range $[0,5]$, Car2's initial position is $60$ (Left), $40$ (Center) and $25$ (Right).}
\label{fig:cars}
\end{figure}

\paragraph{Building heating system.}
Consider a building with $N$ rooms,  each with a heater (\chuchu{see Appendix~\ref{sec:room}} for pseudocode and details). For $i\in[N]$, $x[i] \in \reals$ is the temperature of room $i$ and $m[i] \in\{0,1\}$ captures the off/on state of it's heater. The controller  measures the temperature of rooms periodically; based on these measurements ($y[i]$)  heaters turn on or off. These decisions are made asynchronously across rooms in arbitrary order. The room temperature $x[i]$ changes linearly according to the heater input $m[i]$, the thermal capacity of the room,  and the thermal coupling across adjacent rooms as given in the benchmark problem of~\cite{fehnker2004benchmarks}. For $i \in [N]$, actions $\auto{on}_i,\auto{off}_i$ capture the decision making process of room $i$ on whether or not to turn on the heater. 
Time elapse is captured by a  $\act{flow}$ action that updates the temperatures.
%The boolean variable $d_i$ indicates whether room $i$ has made a decision. 
We want to verify that the room temperatures remain in the $[60, 79]$ range.

%\sayan{Rewire the second para for this. What are the main points? Give counts of actual executions, savings, any insights?}
Consider a building with   $N=3$ rooms. 
\chuchu{In Appendix~\ref{sec:room}, we provide computation details to show that for any $i,j\in[3]$ with $i\neq j$, $a\in\{\act{on}_i,\act{off}_{i}\}$ and $b\in\{\act{on}_j,\act{off}_j\}$, $a\mysim \varepsilon b$ with $\varepsilon = 0.6$; but, $\act{flow}$ is not independent with any other actions.}
Computed reachset over-approximation for 8 rounds and  projected on the temperature of Room 0 is shown in  Figure~\ref{fig:linear_ex}~({\em Right}). Indeed, temperature of Room 0 is contained within the range.

For a round, where each room makes a decision once in arbitrary order, there are  $3! = 6$ $\varepsilon$-equivalent action sequences.
Therefore, from a single initial state, there are $6^8$ (1.6 million) valid executions. Algorithm~\ref{alg:main}, in this case explore only one (length $32$) execution with $\varepsilon = 0.6$ to approximate all executions starting from an initial set with radius $\delta = 2$.
The running time for Algorithm \ref{alg:main} is 1 millisecond while exploring all valid executions from even only a single  state took 434 seconds.

\section{Conclusion}
We proposed a partial order reduction technique for reachability analysis of infinite state transition systems that exploits approximate independence and bounded sensitivity of actions to reduce the  number of executions explored.
This relies on a novel  notion of $\varepsilon$-independence that generalizes the traditional notion of independence by allowing approximate commutation of actions. 
With this $\varepsilon$-independence relation, 
%we are able to define an $\varepsilon$-equivalence of traces and 
we have developed an algorithm for  soundly over-approximating reachsets of all executions using only  $\varepsilon$-equivalent traces. 
The over-approximation can also be made arbitrarily precise by reducing the size of $\delta,\varepsilon$.
%
%\sayan{Is the algorithm still $n^3$?}
In experimental evaluation with three case studies we observe that it can reduce  the number of executions explored exponentially compared to explicit computation of all executions.

%This $O(n^3)$  algorithm can potentially compute the reachset for $O(n!)$ distinct traces.
%We applied the algorithm to verify properties of a linear transition system and a heater control system.
%
The results suggest several  future research directions. In Definition~\ref{por:def:appind}, $\varepsilon$-independent actions are required to be approximately commutative globally.
%That is, from {\em any} state $q$, the resulting states after executing $\varepsilon$-independent actions in any order are close to each other.
For reachability analysis, this  definition could be relaxed to actions that approximately commute locally over parts of the state space.  
%A logical next step is to extend the notion of $\varepsilon$-independent actions to capture locally independent action pairs.
%%%
An orthogonal  direction is to apply this reduction technique to verify temporal logic properties and extend it to hybrid models.

\bibliographystyle{splncs03}
\bibliography{por,iss,dp,sayan1}

\appendix
%% !TEX root = main.tex
\section{Appendix}

\subsection{Algorithm to compute the earliest equivalent point}
\label{app:eep}
In the following algorithm, we find the earliest equivalent point $eep$ of action $a$ on an action sequence $\tau$.
For any trace $\tau$ and action $a$, $eep(\tau,a,\varepsilon)$ constructs a trace $\phi\in A^*$.
Initially $\phi$ is set to be the empty sequence. 
Iteratively, from the end of $\tau$,  we add action $\tau(t)$ to $\phi$ if it is not independent to the entire trace $\phi a$.
We will prove that, length of $\phi$ gives the $eep$ of action $a$ on trace $\tau$.
The time complexity of the algorithm is at most $O(n^2)$, where $n$ is the length of trace $\tau$.

\begin{algorithm}
\caption{eep($\tau, a, \varepsilon$)}
\label{por:algo:anchor}
\begin{algorithmic}[1]
%\Procedure{Bloat}{$\xi, \delta_0, \varepsilon, \beta$}
   \State $\phi \gets \langle \rangle$;
   \State  $T \gets len(\tau)$;
   \For{$t = T-1 : 0$}
%   	    //{\em let $\tau_t$ be the $t$-th action on $\tau$}
   	   \If{$\exists b \in \phi a,\tau(t) \not\mysim{\varepsilon} b$ \label{por:ln:anchor1}} 
		\State $\phi\gets \tau(t)\phi$; \label{por:ln:anchor2}
		%\Comment{$a$ cannot be swapped}
	   \EndIf
   \EndFor
   \State \Return $len(\phi)$;
%\EndProcedure  
\end{algorithmic}
\end{algorithm}

\begin{lemma}
\label{por:lem:anchor}
For any action $a\in A$ and trace $\tau\in A^*$, the function $eep(\tau, a, \varepsilon)$ computes the $eep$ of $a$ on $\tau$.
\end{lemma}
\begin{proof}
For a trace $\tau$ and an action $a$, algorithm $eep(\tau,a, \varepsilon)$ constructs a trace $\phi$ and returns its length. 
%Let $p = \min_{\phi' a\eta' \myequiv\varepsilon \tau a, a\notin \eta'} len(\phi')$ be the $eep$ of $a$ on $\tau$.
To prove that $len(\phi)$ gives the $eep$ $k$ of $a$ on $\tau$, we show both $len(\phi) \geq k$ and $len(\phi) \leq k$.
 
\noindent
{\bf \boldmath $len(\phi) \geq k$}: It suffice to prove the statement by constructing a trace $\eta$ such that $\phi a \eta \myequiv\varepsilon \tau a$ and $a\notin \eta$.
Let $\eta = \tau \backslash \phi$ be the remaining subsequence of $\tau$ after removing the actions in $\phi$. We note that the ordering of actions in $\eta$ is the same as that in $\tau$.  For each action $c\in \eta$, line~\ref{por:ln:anchor2} is not executed. Hence, for all actions $b\in \phi a$ which is originally to the right of $c$, we have $b\mysim\varepsilon c$. 
Therefore, action $c$ can be swapped repeatedly to the right of action $a$. 
Repeat this process for all actions in $\eta$, we derive trace $\phi a \eta$ from the original trace $\tau a$. Therefore $\phi a \eta \myequiv \varepsilon \tau a$.
In addition, we note that from Definition~\ref{por:def:appind}, an $\varepsilon$-independent action pair consists of two distinctive actions, which implies $a\not\mysim\varepsilon a$.
Hence, for each occurrence $a\in \tau$, line~\ref{por:ln:anchor2} is not executed, that is,
$a\notin \eta$.
Therefore, the statement holds.

\noindent
{\bf \boldmath $len(\phi) \leq k$}: 
First, we convert any trace $\tau a$ to a trace consists of only distinctive actions.
If otherwise some action $b\in \tau a$ occurs more than once, we replace the occurrences as distinctive pseudo-actions $b_0,b_1,\dots$, such that each $b_i$ inherit the same independence relation from $b$ and any pair of these pseudo-actions is not independent.
In this way, we map an arbitrary trace $\tau a$ to a trace consists of only distinctive actions. It can be checked that this mapping is bijective.
Without loss of generality, we assume that the actions in $\tau a$ are distinctive. \\
We prove $len(\phi) \leq k$ by contradiction. Suppose $len(\phi) > k$, then there exist traces $\phi',\eta'$ such that (i) $a\notin \eta'$, (ii) $\phi' a \eta' \myequiv\varepsilon \tau a$, and (iii) $len(\phi') < len(\phi)$.
From (iii), there exists an action $c \in \phi \backslash \phi'$. If there are multiple choices of such actions, we choose the rightmost action $c$ in $\phi$.
From line~\ref{por:ln:anchor1} and~\ref{por:ln:anchor2}, action $c$ is in $\phi$ iff there exists another action $b\in\phi$ to the right of $c$ such that $c\not\mysim\varepsilon b$.
Since we choose action $c$ as the rightmost action in $\phi$ that is not in $\phi'$,
we have $b\in\phi'$. 
Originally in trace $\tau a$, action $b$ is to the right of action $c$.
As actions $b$ and $c$ are not  $\varepsilon$-independent, in any equivalent trace $\phi' a \eta' \myequiv\varepsilon \tau a$, the relative position of them should not be changed. 
Hence in trace $\phi' a \eta'$,  action $b$ is also to the right of action $c$. 
However, since $b\in \phi'$ and $c\notin\phi'$, we have action $c$ is to the right of action $b$ in trace $\phi' a \eta'$.
We derive a contradiction. Therefore, if the actions in $\tau a$ are distinctive, $len(\phi) \leq k$. \\
\end{proof}

\subsection{Complete description of the examples}
\subsubsection{Platoon.}
\label{app:platoon}
Consider an $N$ car platoon on a single lane road (see Figure~\ref{por:hioa:car}). Each car can choose one of three actions at each time step: $a$ (accelerate), $b$ (brake), or $c$ (cruise). Car $0$ can choose any action at each time step; remaining cars try to keep safe distance with predecessor by choosing accelerate ($a$) if the distance is more than $50$,  brake ($b$) if the distance is less than $30$, and cruise ($c$) otherwise.
%The state variable will have a real-valued part $x$ and discrete-valued part $m$. 
%
For each $i \in [N]$, $x[2i]$ is the position,  $x[2i +1]$ is the velocity, and $m[i]$ is the chosen action,  of the $i^{th}$ car. 
At each step, $m[i]$ is updated using relative positions according to the rule described above, and then $x$ is updated according to the actions.
For concreteness, the linear state transition equation for a 2-car platoon is shown below:
\begin{figure}[ht]
\centering
  \hrule
  \two{.47}{.47}
  {
\begin{lstlisting}[language=ioaNums,lastline=9]
automaton $\auto{Car Platoon}(N : \ioaNat)$
	variables
		$x:\Real^{2N}$ ;
		$m: \{c,a,b\}^{N}$;
	
	initially
		$\mbox{for each } i \in [N]$
			$m[i]:=$ choose $\{c,a,b\}$;
		
\end{lstlisting}}
  {
\begin{lstlisting}[language=ioaNumsRight,firstline=10]
	transitions
		$\act{move}$
			pre 	true
			eff 	$m[0]:=$ choose $\{c,a,b\}$;
				$\mbox{for each } i \in [N] \setminus 0$
					$m[i]:= \left\{\begin{array}{ll} a & \mbox{ if } x[2(i-1)] - x[2i] >50\\ 
					b & \mbox{ if } x[2(i-1)] - x[2i] <30\\ c & \mbox{ else } \\ \end{array}\right.$;  
					$x:= Ax +  b_m$;


\end{lstlisting}}
  \hrule
  \caption{\small Labeled transition system model of cars keeping a platoon.}
  \label{por:hioa:car}
\end{figure}

\begin{align}
\small
x \leftarrow
\begin{bmatrix}
1 & \Delta t & 0 & 0\\
0 & 1 & 0 & 0 \\
0 & 0 & 1 & \Delta t \\
0 & 0 & 0 & 1 \\ 
\end{bmatrix} x+ 
\begin{bmatrix}
\frac{acc_0 (\Delta t)^2}{2} \\
acc_0  \Delta t \\
\frac{acc_1 (\Delta t)^2} {2} \\
acc_1 \Delta t \\
\end{bmatrix} = Ax + b_{m},
\end{align}
%\sayan{The numbers 0.1, 10 etc. here look a bit hacky. Can we replace $0.1$ with $\delta t$ and $0.01$ with $\delta t^2$? and update the following discussion to be more symbolic? In the last sentence where we talk about $\varepsilon$ values, at that stage we may fix numbers concretely as an example.}
where $acc_{i} >0$ if car $i$ accelerates; $acc_{i}<0$ if it brakes; and $acc_{i}=0$ if it cruises. 
%
%If we fix the value of $acc_{i}$ to be $10$ for accelerating and $-10$ for braking, together $m$ (and $b_m$) could take $9$ different values, so there could be $9$ different actions. 
%
For any value of $m$, the discrepancy function for the corresponding actions are the same: For any $q, q'$ with $q.L = q'.L$, $\beta_a(|q.x - q'.x|) = |A| |q.x - q'.x|$. For any $i,j \in [9]$ with $i \neq j$, we notice that $|a_i a_j(q).x - a_j a_i(q).x| = |Ab_{m_i} - Ab_{m_j} + b_{m_j} - b_{m_i}|$ which is a constant number and can be used as $\varepsilon$. If we choose $\Delta t = 0.1$, then the discrepancy function could be $\beta_a(|q.x - q'.x|_2) =1.06 |q.x - q'.x|_2$. Furthermore, if $acc_{i}$ can choose from $\{-10, 0\}$, or from $\{10, 0\}$, then the corresponding actions are $\varepsilon$-independent with $\epsilon = 0.141$, and if $acc_{i}$ can choose from $\{-10, 0, 10\}$, then the corresponding actions are $\varepsilon$-independent with $\epsilon = 0.282$.
%\sayan{1. instead of moveforward can we not just write x = Ax + bm in the code? 2. Also, is the careful initialization of m[i] important? Can we make the spec shorter by assigning arbitrary m[i] initially for all the cars? 3. I am dropping the initially clause for x[i] as it does not convey any information. }

\subsubsection{Room heating problem}
\label{sec:room}
We present a building heating system in Fig.~\ref{hioa:roomheating}.
The building has $N$ rooms each with a heater.
For $i\in[N]$, $x[i] \in \reals$ is the temperature of room $i$ and $m[i] \in\{0,1\}$ captures the off/on state of the heater in the room.
%Each room switches its heater on and off to control the temperature of the room.
The building measures the temperature of rooms periodically every $T$ seconds and save the measurements to $y[i]$.
Based on the measurement $y[i]$, each room takes action $a_i$ to decide whether to turn on or turn off its heater. The boolean variable $d[i]$ indicates whether room $i$ has made a decision. 
These decisions are made asynchronously among the rooms with a small delay $h$.
For this system, we want to check whether the temperature of the room remains in an appropriate range.

\begin{figure}[ht]
\centering
  \hrule
  \two{.47}{.47}
  {
\begin{lstlisting}[language=ioaNums,lastline=12]
automaton $\auto{Roomheating}(N : \ioaNat)$
	variables
		$x:\Real^N$ initially $x[i]:=60$;
		$y:\Real^N$ initially $y[i]:= 60$;
		$d : \Bool^N$ initially $d :=  \false^N$;
		$m: \Bool^N$  initially $m :=  \false^N$;
		
	transitions
		$\act{on_i}$, for $i\in [N]$
			pre 	$! d[i] \wedge y[i] <= 72$
			eff 	$x:= W_hx + b_h +  C_h m$; 
					$d[i] := \true \wedge m[i] := \true$;
\end{lstlisting}}
  {
\begin{lstlisting}[language=ioaNumsRight,firstline=13]

		$\act{off_i}$, for $i\in [N]$
			pre 	$! d[i] \wedge y[i]  >= 68$
			eff 	$x:= W_hx + b_h +  C_h m$; 
					$d[i] := \true \wedge m[i] := \false$;

						
		$\act{flow}$
			pre	$\wedge_{i\in [N]} d_i $
			eff	$x := W_T x + b_T + C_T m$;
					$d[i] := \false$ for each $i\in[N]$;
					$y:= x$;

\end{lstlisting}}
  \hrule
  \caption{Transition system of room heating.}
  \label{hioa:roomheating}
\end{figure}

For $i\in [N]$, actions $\mathit{on}_i,\mathit{off}_i$ capture the decision making process of room $i$ on whether or not to turn on the heater. 
During the process, time elapses for a (short) period $h$, which leads to an update of the temperature 
as an affine function of current temperature $x$ and the heaters state $m$. 
The affine function is derived from the thermal equations presented in~\cite{fehnker2004benchmarks}. 
In this section, we use an instance of the system with the following matrices:
\begin{equation}
\small
\label{por:eq:heat-a}
W_h = 
\begin{bmatrix}
0.96 &  0.01 & 0.01 \\
0.02 &  0.97 &  0.01\\
0 &  0.01 &  0.97
\end{bmatrix},
b_h = 
\begin{bmatrix}
1.2\\
0\\
1.2
\end{bmatrix}, 
C_h = 
\begin{bmatrix}
0.4 & 0& 0\\
0 & 0 &0\\
0 & 0 & 0.4
\end{bmatrix}.
\end{equation}
After a room controller makes a decision ($\mathit{on}_{i}$ or $\mathit{off}_i$ transition occurs), the variable $d[i]$ to $\true$.
After all rooms make their decisions, action $flow$ captures the time elapse for a (longer) period $T$ which also updates the measured values $y$.
We use an instance of this step with the following matrices:
\begin{equation}
\label{por:eq:heat-a}
\small
W_T = 
\begin{bmatrix}
0.18 & 0.11 &  0.14\\
0.18 &  0.25 &  0.17\\
0.09 & 0.13 & 0.28
\end{bmatrix},
b_T = 
\begin{bmatrix}
34.2 \\
24 \\ 
30
\end{bmatrix},
C_T = 
\begin{bmatrix}
11.4 & 0& 0\\
0 & 8&0\\
0 & 0 & 10
\end{bmatrix}.
\end{equation}

For each $i\in[N]$ and $a_i \in\{\mathit{on}_i,\mathit{off}_i\}$, we will derive the discrepancy function for action $a$.
For any $q, q'$ with $q.L = q'.L$,
\[
\begin{array}{rl}
&|a_i(q).x - a_i(q').x| \\ 
=& |W_h q.x + b_h + C_h q.m - W_h q'.x - b_h - C_h q'.m| \\
\leq &|W_h||q.x-q'.x| 
\end{array}
\]
We note that $|W_h|_2 = 0.99$.
Hence, we can define $\beta_a (|q.x-q'.x|_2) = 0.99|q.x-q'.x|_2$ 
as the discrepancy functions of each $a\in \{\mathit{on}_i,\mathit{off}_i\}_{i\in[3]}$.
Similarly, we derived that $\beta_{flow}(|q.x-q'.x|_2) = 0.52|q.x-q'.x|_2$.

For any $i,j\in[3]$ with $i\neq j$, $a_i \in\{\mathit{on}_i,\mathit{off}_{i}\}$ and $a_j \in\{\mathit{on}_j,\mathit{off}_j\}$, we can prove $a_i \mysim \varepsilon a_j$ with $\varepsilon = 0.6$.
Notice that, $a_i(q).x = W_h q.x+b_h + C_h q.m = a_j(q).x$ are identical, but $a_i(q).m$ and $a_j(q).m$ could be different.
\[
\begin{array}{rl}
&|a_ia_j(q).x - a_ja_i(q).x|  \\
= &|W_h a_j(q).x + b_h + C_h a_j(q).m - W_h a_i(q).x - b_h - C_ha_i(q).m| \\
= & |C_h a_j(q).m - C_h a_i(q).m| \leq |C_h| |a_j(q).m-a_i(q).m|
\end{array}
\]
We note that $|C_h|_2 = 0.4$. We will give an upper bound on $|a_j(q).m-a_i(q).m|$.
Notice that $a_i(q).m$ and $q.m$ can only differ in one bit ($m_i$). 
Similarly, $a_j(q).m$ and $q.m$ can only differ in one bit ($m_j$). 
Hence $a_i(q).m$ and $a_j(q).m$ can be differ in at most two bits, and $|a_i(q).m-a_j(q).m|_2 \leq |[1,1,0]|_2 = 1.41$. 
Therefore, 
\[
|a_ia_j(q).x - a_ja_i(q).x|_2 \leq 0.4*1.41 \leq 0.6.
\]
Thus for any pair of rooms, the on/off decisions are $\varepsilon$-approximately independent with $\varepsilon=0.6$.
For a round, where each room makes a decision once in arbitrary order, there are in total $3! = 6$ $\varepsilon$-equivalent action sequences.

%We present an execution $\xi$ in Figure~\ref{fig:heat_sim} as the blue curve, which ran for 8 rounds. 
%There are a number of 1.6 million ($6^8$) ($\delta, \varepsilon$)-close potential executions of $\xi$ with $\varepsilon = 0.6$ and $\delta = 2$.
%We illustrate a tube computed by our technique in green, which is proved to contain the final state of all these ($\delta,\varepsilon$)-close potential executions.

%\[\tau = on_2,on_0,on_1,\mathit{flow},\mathit{off_1},\mathit{off}_2,\mathit{off}_0,\mathit{flow},on_0,\mathit{off}_1,on_2,\mathit{flow},\mathit{off}_2,on_1,\mathit{off}_0,\mathit{flow},\mathit{off}_1,on_2,on_0,\mathit{flow},on_1,\mathit{off}_0,\mathit{off}_2,\mathit{flow},\mathit{off}_1,on_0,on_2,\mathit{flow}\]

%\begin{figure}[ht!]
%\centering
%\includegraphics[width=0.7\textwidth]{room_sim.png}
%\caption{A tube around an execution of a room heating problem. The blue curve is the execution $\xi$. The green curves illustrate a tube around $\xi$ over-approximating ($\delta,\varepsilon$)-close potential execution with $\delta=2$ and $\varepsilon=0.6$.}
%\label{fig:heat_sim}
%\end{figure}

\end{document}